\documentclass[11pt]{article}

\newif\ifanonymous
\anonymousfalse

\newif\ifreview

\reviewfalse
\usepackage[utf8]{inputenc}
\usepackage[T1]{fontenc}
\usepackage{lmodern}

\usepackage{amsmath, amssymb, amsthm, amsfonts}
\usepackage{mathtools}
\usepackage{breqn}

\usepackage[margin=1in]{geometry}
\usepackage{setspace}
\usepackage{lineno}
\ifreview
  \linenumbers                       
\else
\fi
\usepackage{enumitem}
\usepackage{appendix}
\usepackage{authblk}
\usepackage{graphicx}
\usepackage{booktabs}                
\usepackage{caption}
\usepackage{subcaption}
\usepackage{float}

\usepackage{natbib}
\bibpunct{(}{)}{;}{a}{,}{,}

\usepackage[colorlinks=true,allcolors=blue]{hyperref}
\usepackage{cleveref}

\newtheorem{theorem}{Theorem}[section]
\newtheorem{lemma}[theorem]{Lemma}
\newtheorem{proposition}[theorem]{Proposition}

\theoremstyle{definition}
\newtheorem{definition}[theorem]{Definition}

\theoremstyle{remark}
\newtheorem{remark}[theorem]{Remark}

\begin{document}

\title{\textbf{On regularity of finite-maturity American put options in the Heston model}}

\ifanonymous
  \author{}
\else
  \author[1]{Khai Nguyen\thanks{Email: \texttt{hungquangkhai.nguyen@postgrad.manchester.ac.uk}}}
  \author[1]{Huy Chau\thanks{Email: \texttt{huy.chau@manchester.ac.uk}}}
  \affil[1]{Department of Mathematics, The University of Manchester, Manchester, United Kingdom}
\fi

\date{\today}
\maketitle

\begin{abstract}
This paper studies the regularity of finite-maturity American value functions in the Heston model. Although the Heston operator is degenerate when the volatility is zero, we are able to establish \( C^{1,2} \) regularity of the American value functions in the exercise domain and the smooth-fit principle, using PDE techniques.
\end{abstract}

\\[0.3em]
\\[0.3em]

\vspace{1em}


\section{Introduction}\label{sec:introduction}

American options give holders the right to exercise at any time up to expiration, a flexibility that distinguishes them from their European counterparts and introduces fundamental challenges in pricing, hedging, and regularity analysis. The early-exercise feature converts the pricing problem into an optimal stopping problem: the American option value equals the supremum, over all admissible stopping times, of the expected discounted payoff under the risk-neutral measure. The mathematical foundations of this theory were laid in a series of classical contributions. \citet{McKean1965} was the first to reformulate American option pricing as a free-boundary problem, converting the optimal stopping formulation into a partial differential equation complemented by a moving boundary that separates the continuation region from the exercise region. \citet{Merton1973} developed the no-arbitrage framework for rational option pricing and demonstrated that American calls on non-dividend-paying stocks coincide with their European counterparts, while American puts generally carry an early-exercise premium. \citet{BensoussanLions1982} introduced the variational-inequality approach to
stochastic control and optimal stopping, providing a unified analytic framework
that would become central to the rigorous study of American option pricing. \citet{peskir2006} gave a comprehensive probabilistic treatment of optimal stopping and free-boundary problems, rigorously establishing the smooth-fit principle and related regularity results for a broad class of diffusions.

The optimal exercise boundary (or free boundary) is the central object of the theory: it is the critical surface in the state space that divides the continuation region---where it is optimal to hold the option---from the exercise region---where immediate exercise is optimal. Knowledge of this boundary completely determines the early-exercise strategy and the option price. Crucially, to derive an analytic characterisation of the boundary---typically an integral equation obtained by applying It\^o's formula to the value function---one needs both the smooth-fit principle (continuity of the value function and its first-order spatial derivative across the boundary) and $C^{1,2}$ regularity of the value function in the continuation region. Without these properties, It\^o's formula cannot be applied up to the free boundary, and the integral equation that identifies the exercise boundary cannot be obtained. The smooth-fit condition is therefore not merely a technical regularity statement; it is indispensable for the derivation of early-exercise premium representations \citet{Carr1992, KimIE1990, jacka1991}.

In the classical Black--Scholes model \citet{Black1973}, all the desired regularity properties are known to hold. The American put value function is continuous on $[0,T]\times(0,\infty)$, belongs to $C^{1,2}$ inside the continuation region where it solves the Black--Scholes PDE, and satisfies the smooth-fit principle at the free boundary. Moreover, the exercise boundary itself is continuous, monotone and unique. These results were established and refined in the subsequent literature \citet{jacka1991, lamberton2013, chen2012, Friedman1975, peskir2006}. However, the constant-volatility assumption of the Black--Scholes model is inconsistent with empirical features of option markets such as the volatility smile and skew \citet{Hull1987}, motivating the development of stochastic volatility models; among these, the Heston model \citet{Heston1993}, in which the variance follows a Cox--Ingersoll--Ross process \citet{CIR1985}, has become one of the most widely used owing to its analytical tractability for European options. While the pricing and calibration of European options under Heston dynamics are by now well understood, the regularity of American option prices in this framework remains  less developed, since the integral equation representations and numerical schemes used to compute American option prices and exercise boundaries cannot be rigorously justified without such regularity results.

The regularity analysis of finite-maturity American put options under the Heston
stochastic-volatility dynamics is considerably more challenging than in the
Black--Scholes framework. The pricing problem is translated into a free-boundary
(obstacle) problem for a degenerate parabolic operator on the unbounded state space
$(0,T]\times\mathbb{R}\times\mathbb{R}^+$ (time, log-price, variance). The Heston
generator degenerates along the boundary $\{y = 0\}$, where the variance vanishes;
this lack of uniform ellipticity rules out the direct application of classical
parabolic PDE theory \citet{Ladyzenskaja1968, Lieberman1996, Friedman1964}. When the operator is uniformly parabolic, the regularity theory for American option
variational inequalities is well developed. The variational-inequality approach to American option pricing, yielding
existence, uniqueness, and basic regularity of the value function, was
established by \citet{Jaillet1990}, building on the framework of
\citet{BensoussanLions1982}. For the regularity theory of the obstacle problem under uniformly elliptic/ parabolic operators, we refer to \citet{Friedman1982,KinderlehrerStampacchia1980, Caffarelli1998}. Building on these foundations, \citet{Pascucci2011} gave a comprehensive treatment of the penalty method approach
in Chapter~8, obtaining $W^{2,1}_{p,\mathrm{loc}}$ regularity and smooth-fit for
obstacle problems governed by uniformly parabolic operators. In a
different direction, \citet{huyenpham} and \citet{haoxing2012regularity} studied the
regularity of optimal stopping problems for jump diffusions using viscosity solution
techniques; in particular, \citet{haoxing2012regularity} established both the
smooth-fit principle and $C^{1,2}$ regularity in the continuation region under the
assumption that the operator is uniformly elliptic throughout the domain. All of
these frameworks, however, require uniform parabolicity or ellipticity on the full
state space, an assumption that fails for the Heston generator due to its
degeneracy at $\{y=0\}$. Consequently, their results do not directly apply to the
present setting.

Significant progress has been made for degenerate operators sharing the same structural degeneracy as the Heston generator, particularly in the \emph{elliptic} setting, which corresponds to the pricing of perpetual (infinite-maturity) American options. \citet{Daskalopoulos2013} proved existence, uniqueness, and global regularity of solutions for degenerate elliptic obstacle problems with direct applications to mathematical finance. Building on this, \citet{Daskalopoulos2012} obtained $C^{1,1}$ regularity for the same class of obstacle problems, which represents the optimal regularity for the value function of the perpetual American put under Heston dynamics. \citet{feehan2015degenerateelliptic} established higher-order weighted Sobolev and H\"older regularity for solutions of variational equations governed by the \emph{elliptic} Heston operator; moreover, they showed that when the data are $C^\infty$-smooth, the associated value functions inherit $C^\infty$ regularity all the way up to the portion of the boundary where the operator degenerates. These results provide a thorough understanding of the regularity landscape in the perpetual (time-independent) case. However, they do not directly extend to the finite-maturity (parabolic) setting, because the additional time variable introduces new analytical difficulties, including the interaction between temporal regularity and the spatial degeneracy of the operator.

For the \emph{parabolic} (finite-maturity) problem, the available results are more limited. In a related direction, \citet{Feehan2014} proved existence and uniqueness of stochastic representations for solutions to both elliptic and parabolic boundary value and obstacle problems driven by degenerate operators of Heston type. Their framework encompasses the parabolic case; however, it remains unclear whether the assumptions they impose on the solution of the parabolic obstacle problem are satisfied by the American option value function with the standard put payoff. \citet{Lamberton2019} employed semigroup techniques to study the finite-maturity American option price in Heston-type models. Working within carefully chosen weighted Sobolev spaces, they established continuity of the price function together with weak regularity results (specifically, membership in weighted $L^2$-based Sobolev spaces). Their work represents an important step, but the regularity they obtained falls well short of the classical $C^{1,2}$ smoothness needed for the smooth-fit principle and the application of It\^o's formula at the free boundary. In a complementary direction, \citet{LambertonTerenzi2019properties} studied the American put
price in the Heston model using probabilistic techniques and proved a \emph{weak} form of the
smooth-fit property. Their result shows that, for each fixed time and volatility level, the
derivative of the value function with respect to the asset price, computed as a one-sided limit
from the continuation region toward the free boundary, coincides with the corresponding derivative
of the payoff function; an analogous statement holds for the derivative with respect to the
volatility variable under the Feller condition. However, this weak smooth-fit is established only
\emph{pointwise} at the boundary along a single direction of approach, namely by varying the
asset price while holding time and volatility fixed, and does not imply that the spatial derivatives are continuous up to and across the free boundary. Without this global continuity of the spatial derivatives, one cannot apply It\^o's formula across the free boundary to derive the integral equation characterising the exercise boundary. Moreover, their work did not establish classical regularity of the value function in the continuation region.

The $C^{1,2}$ regularity of the American put value function and the smooth-fit principle under the Heston model on the full unbounded domain $\mathbb{R}\times\mathbb{R}^+$ have not been established for the finite-maturity case. The existing results either address only the elliptic (perpetual) problem \citet{Daskalopoulos2013, Daskalopoulos2012, feehan2015degenerateelliptic}, impose assumptions that are not verified for the standard American put \citet{Feehan2014}, or yield only weak regularity insufficient for the smooth-fit principle \citet{Lamberton2019}. Without $C^{1,2}$ regularity and smooth-fit, one cannot rigorously derive the early-exercise premium representation, justify the convergence of standard numerical schemes, or characterise the exercise boundary via an integral equation.

\textbf{Our contribution.} This paper addresses the regularity of finite-maturity American put option prices in the Heston model over the full unbounded domain $(0,T]\times\mathbb{R}\times\mathbb{R}^+$. Our main results can be summarised as follows.
\begin{enumerate}[label=(\roman*)]
  \item \emph{Smooth-fit principle and $C^{1,2}$ regularity.} We prove that the American put
  value function $P^A$ satisfies the smooth-fit principle---that is, its first-order spatial
  derivatives are continuous across the free boundary---and that $P^A$ belongs to $C^{1,2}$ in the
  exercise domain. This extends the elliptic regularity results of
  \citet{Daskalopoulos2013, Daskalopoulos2012} and \citet{feehan2015degenerateelliptic} to the
  parabolic (finite-maturity) setting. Our result also strengthens
  the weak smooth-fit property of \citet{LambertonTerenzi2019properties} to the \emph{strong} smooth-fit principle, in which the first-order spatial derivatives
  are continuous functions on the entire state space. This global continuity of the spatial
  derivatives across the free boundary is precisely the
  condition needed to apply It\^o's formula across the free boundary and derive the integral
  equation for the exercise boundary.

  \item \emph{$W^{2,1}_{p,\mathrm{loc}}$ regularity.} We show that the solution to the
  associated variational inequality belongs to the local Sobolev space
  $W^{2,1}_{p,\mathrm{loc}}(E_T)$ for every $p \geq 3$. This result is new for the
  finite-maturity American option under the Heston model and is strictly stronger than
  the regularity obtained in prior work. Specifically, \citet{Lamberton2019} showed that the American put price function belongs to weighted $L^2$-based Sobolev spaces using
  semigroup techniques, but their result does not yield $W^{2,1}_p$ regularity for
  $p \geq 3$, which is the range needed for the Sobolev embedding theorem to produce
  continuous first-order derivatives. In the elliptic (perpetual) setting, \citet{feehan2015degenerateelliptic} obtained higher-order weighted Sobolev and
  H\"older regularity for variational equations governed by the elliptic Heston operator,
  but their results do not extend to the parabolic case. For uniformly
  parabolic operators, $W^{2,1}_{p,\mathrm{loc}}$ regularity of obstacle problems is
  classical and follows from the general theory of \citet{Ladyzenskaja1968} and
  \citet{Friedman1964}; it is the standard tool used in \citet{haoxing2012regularity} and
  \citet{Pascucci2011} to deduce smooth fit via Sobolev embedding. However, these
  classical results require uniform parabolicity, which fails for the Heston operator.
  Our contribution is to establish the same $W^{2,1}_{p,\mathrm{loc}}$ regularity in
  the degenerate Heston setting, thereby providing the analytical foundation from which
  the smooth-fit principle is deduced via the Sobolev embedding theorem, exactly as in
  the uniformly parabolic case.

  \item \emph{Uniqueness of the viscosity solution.} We establish that the American put
  value function $P^A$ is the unique viscosity solution of the associated variational
  inequality. This result is not immediate in the Heston setting because the coefficients
  of the Heston stochastic differential equation do not satisfy the global Lipschitz
  condition, which prevents the direct application of the classical comparison principle
  for viscosity solutions of optimal stopping problems (cf.\ \citet[Theorem~4.1]{huyenpham}).
  To overcome this, we adapt the comparison argument developed in \citet{ishii2021existence}
  for integro-differential equations arising in option pricing. Their framework, however,
  is formulated for models in which the volatility process is driven by a pure drift
  component together with a L\'evy process, which does not completely coincide with the Heston
  dynamics where the volatility is driven by a diffusion with a square-root coefficient.
  We therefore modify their argument to accommodate the specific degeneracy and diffusion
  structure of the Heston model. The details of this adaptation are given in
  Proposition~\ref{lem:viscosity_solution}.

  \item \emph{Proof technique.} The Heston operator degenerates along the boundary
  $\{y=0\}$ of the state space $\mathbb{R}\times\mathbb{R}^+$, which places the problem
  outside the scope of several well-established approaches to American option regularity.
  We briefly compare our strategy with three representative frameworks in order to
  highlight the specific difficulties introduced by the degeneracy and the techniques we
  develop to overcome them.

  \begin{itemize}
    \item
    In the classical Black--Scholes setting, \citet{peskir2006} establish the smooth-fit
    principle by a direct argument based on optimal stopping times and the continuity
    definition of the derivative at the free boundary. A key prerequisite of their approach
    is that $C^{1,2}$ regularity of the value function inside the continuation region is
    already available, which in the Black--Scholes case follows immediately from standard
    interior regularity for uniformly parabolic PDEs. In the Heston model, interior
    regularity does not follow from classical theory due to the degeneracy of the operator,
    and consequently the Peskir--Shiryaev argument cannot be applied to obtain the strong
    smooth-fit principle. Indeed, \citet{LambertonTerenzi2019properties} adapted the same
    probabilistic technique to the Heston model and obtained only a \emph{weak} smooth fit
    (pointwise matching of one-sided derivatives), precisely because $C^{1,2}$ regularity
    in the continuation region was unavailable.

    \item
    \citet{haoxing2012regularity} study the regularity of optimal stopping problems for jump
    diffusions by formulating a penalised PDE directly on the full state space. Their
    approach requires the operator to be uniformly elliptic throughout the domain, which
    is used to invoke the classical theory of \citet{Ladyzenskaja1968} and to establish the
    existence of a fixed point in their Lemma~4.2. Since the Heston operator lacks uniform
    ellipticity, this method cannot be applied directly to our setting. Despite these differences, our
    work draws on several elements of their framework: we adopt a similar construction of
    the mollified functions, and the technique used to derive uniform bounds on the penalised
    solution in Proposition~\ref{lem:bound_epsilon} follows a comparable strategy.

    \item
    The approach in Chapter~8 of \citet{Pascucci2011} is closest in spirit to ours: the
    obstacle problem is first solved on bounded subdomains via a penalty approximation,
    and regularity is obtained from interior Sobolev estimates before passing to the limit
    on expanding domains. However, the framework of \citet{Pascucci2011} assumes that the
    operator is uniformly parabolic on the entire domain on which the penalised PDE is
    posed, so that the classical interior $W^{2,1}_p$ estimates apply directly and the
    Sobolev bounds transfer immediately to the limit---this is the mechanism underlying the
    proofs of Theorems~8.27 and~8.30 in \citet{Pascucci2011}. In our setting, the
    degeneracy of the Heston operator prevents the direct application of these estimates on
    any domain that touches $\{y = 0\}$. We address this difficulty in two ways. First, we
    develop a covering argument that yields interior $W^{2,1}_p$ estimates on subdomains
    bounded away from the degeneracy locus, with constants that are uniform in the choice
    of subdomain (Proposition~\ref{lem:cauchy_problem_solution}
    and the proof of Proposition~\ref{lem:subsequence_existence}). Second, in the passage from bounded
    subdomains to the full unbounded domain, we adopt a simpler choice of boundary
    function than the one employed in \citet{Pascucci2011}, which streamlines the
    convergence argument (see the proof of
    Proposition~\ref{lem:subsequence_existence}).
\end{itemize}

  Our approach overcomes these obstructions through a three-stage strategy tailored to the
  degenerate structure of the Heston operator. First, we establish that the value function
  $P^A$ is continuous and is the unique viscosity solution of the associated variational
  inequality, with uniqueness following from a comparison principle adapted from
  \citet{ishii2021existence}. Second, we formulate the penalised PDE on bounded subdomains
  $D^n_T$ that are \emph{bounded away from the degeneracy locus $\{y=0\}$}. On each
  such subdomain, the Heston operator is uniformly elliptic, and the classical theory of
  \citet{Ladyzenskaja1968} provides existence, uniqueness, and H\"older regularity of the
  penalised solution (Proposition~\ref{lem:penalty_var}). The key results lie in the third
  stage: we derive pointwise bounds (Proposition~\ref{lem:bounds_v_epsilon}) and interior $W^{2,1}_p$ estimates
  (Proposition~\ref{lem:cauchy_problem_solution} and~\ref{lem:subsequence_existence}) that are \emph{uniform in both the penalty parameter $\epsilon$ and the
  subdomain $D^n$}. In our setting, it is precisely this property that allows us to pass to the double limit $\epsilon\to 0$, $n\to\infty$ and construct a function $v^*\in W^{2,1}_{p,\mathrm{loc}}(E_T)$ that is also a viscosity solution of the variational inequality (Proposition~\ref{lem:subsequence_existence}). Finally, the uniqueness result of viscosity solution and the Sobolev embedding theorem, yields the strong smooth-fit principle for $P^A$
  (Proposition~\ref{lem:smoothfits}).
\end{enumerate}

\textbf{Organisation of the paper.} Section~\ref{sec:model} introduces the Heston model and formulates the American put pricing problem as a variational inequality. Section~\ref{sec:notations_background} collects preliminary results, including the continuity of the value function. Section~\ref{sec:viscosity_solution} establishes that the American put price is the unique viscosity solution of the variational inequality. Section~\ref{sec:results} contains the main regularity results: using penalised approximations, Sobolev-space estimates, and viscosity-solution arguments, we prove the smooth-fit principle and $C^{1,2}$ regularity.

\section{The model}\label{sec:model}
Under a risk-neutral pricing measure, the dynamics of the asset price \( S_t \) and its variance \( Y_t \) in the Heston model are governed by the following system of stochastic differential equations (SDEs):
\begin{equation}
\begin{cases}
\frac{dS_t}{S_t} = (r - \delta) \, dt + \sqrt{Y_t} \, dB_t, & S_0 = s > 0, \\
dY_t = \kappa (\theta - Y_t) \, dt + \sigma \sqrt{Y_t} \, dW_t, & Y_0 = y > 0,
\end{cases}
\end{equation}
where \( r \) is the risk-free interest rate, \( \delta \) is the continuous dividend yield, \( \kappa > 0 \) is the mean reversion rate, \( \theta > 0 \) is the long-run mean of the variance, \( \sigma > 0 \) is the volatility of volatility, \( B_t \) and \( W_t \) are two standard Brownian motions with correlation \( d\langle B, W \rangle_t = \rho \, dt \), \( \rho \in (-1, 1) \).

Throughout this paper, we assume the Feller condition \( 2\kappa\theta \geq \sigma^2 \), ensuring the Cox-Ingersoll-Ross process \( Y_t \) remains strictly positive at all times. We work with the transformed process \( X_t = \log S_t \), adjusting the drift to simplify the Heston operator. The joint dynamics of \( X_t \) and \( Y_t \) are:
\begin{equation}
\label{eq:heston_stochastic_equation}
\begin{cases}
dX_t = \left( r-\delta - \frac{Y_t}{2} \right) \, dt + \sqrt{Y_t} \, dB_t, & X_0 = x, \\
dY_t = \kappa (\theta - Y_t) \, dt + \sigma \sqrt{Y_t} \, dW_t, & Y_0 = y > 0.
\end{cases}
\end{equation}
Under this transformation, an American put option with strike price \( K \) and maturity \( T \) gives its holder the payoff
\begin{equation}
G(x) = \left( K - e^{x} \right)^+,
\end{equation}
where \( (z)^+ = \max(z, 0) \). The value of the American put option at time \( t \), given the state \( (X_t, Y_t) = (x, y) \), is defined as
\begin{equation}
\label{def: PA_american_option}
P^A(t, x, y) = \sup_{\tau \in \mathcal{T}_{t,T}} \mathbb{E} \left[ e^{-r (\tau - t)} G( X_\tau^{t,x,y}) \right],
\end{equation}
where \( \mathcal{T}_{t,T} \) denotes the set of stopping times \( \tau \) taking values in \( [t, T] \), and \( (X_s^{t,x,y}, Y_s^{t,x,y})_{t \leq s \leq T} \) is the solution of \eqref{eq:heston_stochastic_equation} with the initial conditions \( X^{t,x,y}_t = x \) and \( Y^{t,x,y}_t = y \).

In the setting of optimal stopping, we define the continuation region $\mathcal{C}$ and stopping region $\mathcal{D}$ as follows:
\[
\mathcal{C} := \left\{ (t,x,y) \in [0,T] \times \mathbb{R} \times \mathbb{R}^+ : P^{A}(t,x,y) > G(x) \right\}, \quad
\mathcal{D} := \left\{ (t,x,y) \in [0,T] \times \mathbb{R} \times \mathbb{R}^+ : P^{A}(t,x,y) = G(x) \right\}.
\]
The pricing problem involves determining \( P^A(t, x, y) \) and the optimal exercise boundary, where the option holder chooses to exercise early. Assuming sufficient smoothness, \( P^A(t, x, y) \) satisfies the obstacle problem
\begin{equation}
\label{eq:heston_pde}
\begin{cases}
\min \left\{ \left( -\partial_t - \mathcal{L} + r \right) u, u - G \right\} = 0, & (t, x, y) \in [0, T) \times \mathbb{R} \times \mathbb{R}^+, \\
u(T, x, y) = G(x), & (x, y) \in \mathbb{R} \times \mathbb{R}^+,
\end{cases}
\end{equation}
where \( \mathcal{L} \) is the Heston operator, defined as
\begin{equation}
\label{eq:heston_operator}
\mathcal{L} = \frac{y}{2} \left( \frac{\partial^2}{\partial x^2} + 2 \rho \sigma \frac{\partial^2}{\partial x \partial y} + \sigma^2 \frac{\partial^2}{\partial y^2} \right) + \left( r -\delta - \frac{y}{2} \right) \frac{\partial}{\partial x} + \kappa (\theta - y) \frac{\partial}{\partial y}.
\end{equation}
This operator is degenerate when $y = 0$, making it more difficult to employ typical PDE techniques.

\textit{Notations.} For any set \( D \subset \mathbb{R} \times \mathbb{R}^+ \), we define the parabolic cylinder over the time interval \( (0, s] \) as \( D_s = (0, s] \times D \), with its parabolic boundary given by \( \partial_P D_s = \partial D_s \setminus (\{s\} \times D) \). We also introduce the unbounded domain \( E_s = (0, s] \times \mathbb{R} \times \mathbb{R}^+ \). For a vector \( z = (x,y) \in \mathbb{R}^2 \), its Euclidean norm is \( |z| = \sqrt{x^2 + y^2} \). Let \( M_n(\mathbb{R}) \) denote the set of all \( n \times n \) real matrices, and let \( \mathbb{S}_n \) denote the subspace consisting of all \( n \times n \) symmetric matrices. For any matrix \( A \in M_n(\mathbb{R}) \), let \( \operatorname{Tr}(A) \) denote the trace of \( A \), defined as the sum of its diagonal entries. For vectors \( a, b \in \mathbb{R}^n \), we write \( \langle a, b \rangle = \sum_i a_i b_i \) for the Euclidean inner product; for matrices \( A, B \in M_n(\mathbb{R}) \), we define \( \langle A, B \rangle = \operatorname{Tr}(AB^T) \). We denote by \(I_n\) the \( n \times n\) identity matrix. For any \( N \times d \) matrix \( A = (a_{ij}) \), let \( |A| = \sqrt{\sum_{i=1}^{N} \sum_{j=1}^{d} a_{ij}^{2}} \). For any two sets \( \mathcal{O}_1 \) and \( \mathcal{O}_2 \), the notation \( \mathcal{O}_1 \Subset \mathcal{O}_2 \) indicates that \( \overline{\mathcal{O}_1} \subset \mathcal{O}_2 \). For a function \(\phi\), we denote by \(D \phi\) the gradient (i.e. the vector of first-order partial derivatives) of \(\phi\), and by \(D^2 \phi\) its Hessian matrix (the matrix of second-order partial derivatives). Let \( C([0,T] \times \mathbb{R} \times \mathbb{R}^+) \) denote the space of all continuous functions defined on \( [0,T] \times \mathbb{R} \times \mathbb{R}^+ \). The space \( C^{1,2}(D_s) \) consists of all continuous functions on \( D_s \) that possess continuous classical derivatives with respect to time up to the first order and with respect to the spatial variables up to the second order.

The space $L_p(D_s)$ is the Lebesgue space consisting of all measurable functions $v: D_s \to \mathbb{R}$ such that $\int_{D_s} |v(t,z)|^p \, dt\, dz < \infty$, where $1 \leq p < \infty$. The norm is given by
\begin{equation}
   \|v\|_{L_p(D_s)} = \left( \int_{D_s} |v(t,z)|^p \, dt dz \right)^{1/p}.
\end{equation}
For $p = \infty$, $L_\infty(D_s)$ consists of measurable functions $v$ such that $\|v\|_{L_\infty(D_s)} = \operatorname{ess\,sup}_{(t,z) \in D_s} |v(t,z)| < \infty$.

The Sobolev space $W^{2,1}_p(D_s)$ consists of all functions $v \in L_p(D_s)$ such that the generalised derivatives $\partial_t v$, $\partial_x v$, $\partial_y v$, $\partial^2_{xx} v$, $\partial^2_{xy} v$, $\partial^2_{yy} v$ exist and belong to $L_p(D_s)$. It is equipped with the norm
\begin{equation}
\begin{split}
    \|v\|_{W^{2,1}_p(D_s)} = & \|v\|_{L_p(D_s)} + \|\partial_t v\|_{L_p(D_s)} + \|\partial_x v\|_{L_p(D_s)} + \|\partial_y v\|_{L_p(D_s)} \\
    & + \|\partial^2_{xx} v\|_{L_p(D_s)} + \|\partial^2_{xy} v\|_{L_p(D_s)} + \|\partial^2_{yy} v\|_{L_p(D_s)}.
\end{split}
\end{equation}
The space \(W^{2,1}_{p,\mathrm{loc}}(D_s)\) consists of functions whose \(W^{2,1}_{p}\)-norm is finite on every compact subset of \(D_s\).

For any positive real number \(\alpha\), \(H^{\alpha, \alpha/2}(D_s)\) is the space of functions \(v\) that are continuous in \(D_s\) with continuous classical derivatives \(\partial_t^r \partial_z^s v\) for \(2r + s < \alpha\), and have finite norm
\[
\|v\|_{D_s}^{(\alpha)} := |v|_z^{(\alpha)} + |v|_t^{(\alpha/2)} + \sum_{2r+s \leq [\alpha]} \|\partial_t^r \partial_z^s v\|^{(0)},
\]
in which
\[
\|v\|^{(0)} = \max_{D_s} |v|,
\]
\[
|v|_z^{(\alpha)} = \sum_{2r+s=[\alpha]} \sup_{|z-z'| \leq \rho_0} \frac{|\partial_t^r \partial_z^s v(t,z) - \partial_t^r \partial_z^s v(t,z')|}{|z - z'|^{\alpha - [\alpha]}},
\]
and
\[
|v|_t^{(\alpha/2)} = \sum_{\alpha-2 < 2r+s < \alpha} \sup_{|t-t'| \leq \rho_0} \frac{|\partial_t^r \partial_z^s v(t,z) - \partial_t^r \partial_z^s v(t',z)|}{|t - t'|^{(\alpha - 2r - s)/2}},
\]
for a constant \(\rho_0\).

\section{Preliminary results}\label{sec:notations_background}
This section contains standard estimates for the process \(Z^{t,z}_s = (X^{t,x,y}_s,Y^{t,x,y}_s)\), where \((X^{t,x,y}_s,Y^{t,x,y}_s)\) is defined in \eqref{eq:heston_stochastic_equation} with the starting point \(z=(x,y)\), and some continuity properties of the value function $P^A(t,x,y)$.
For simplicity of notation, we define \begin{equation}
\label{eq:def_of_signma_b}
\sigma(t, z) = \begin{bmatrix} \sqrt{y} & 0 \\ \sigma \rho \sqrt{y} & \sigma \sqrt{1 - \rho^2} \sqrt{y} \end{bmatrix}, \quad b(t, z) = \begin{bmatrix} r - \delta - \frac{y}{2} \\ \kappa (\theta - y) \end{bmatrix}.
\end{equation}
The process $Z$ evolves as
$$dZ_t = b(t,Z_t)dt + \sigma(t,Z_t) d\overline{W}_t,$$
where \( d\bar{W}_t = (dB_t, dW^2_t) \), with \( B_t \) and \( W^2_t \) independent, and \( dW_t = \rho \, dB_t + \sqrt{1 - \rho^2} \, dW^2_t \).

\begin{lemma}
\label{lem:stochastic_bound_1}
There exists $C > 0$ such that the following hold.
\begin{enumerate}[label=\alph*)]
    \item For any $k \in [0, 2]$, $h, t \in [0, T]$, $z \in \mathbb{R} \times \mathbb{R}^+$, and $\tau \in \mathcal{T}_{0,h}$:
    \(\mathbb{E} \left| Z^{t, z}_{t+\tau} \right|^k \leq C (1 + |z|^k)\).
    \item For any $k \in [0, 2]$, $h, t \in [0, T]$, $z \in \mathbb{R} \times \mathbb{R}^+$, and $\tau \in \mathcal{T}_{0,h}$:
    \(\mathbb{E} \left| Z^{t, z}_{t+\tau} - z \right|^k \leq C (1 + |z|^k) h^{\frac{k}{2}}\).
    \item For any $h, t \in [0, T]$ and $z, z' \in \mathbb{R} \times \mathbb{R}^+$:
    \(\mathbb{E} \left[ \sup_{0 \leq s \leq h} |Z^{t, z}_{t+s} - Z^{t, z'}_{t+s}| \right] \leq C (|z - z'| + \sqrt{|z-z'|})\).
\end{enumerate}
\end{lemma}
\begin{proof}
By H\"{o}lder's inequality, it suffices to prove these estimates for $k = 2$. For notational simplicity, hereafter, $C$ denotes a generic constant which may vary from line to line.

Consider $(a)$. Using the fact that $b$, $\sigma$ satisfy the linear growth conditions $|b(t,z)|+|\sigma(t,z)| \leq K(1+|z|)$, for some $K>0$, we have
\begin{eqnarray}
\mathbb{E} \left| Z^{t, z}_{t+\tau} \right|^2
&\leq& C \left( |z|^2 + \mathbb{E} \int_0^\tau \left| b(u + t, Z^{t, z}_{t+u}) \right|^2 \, du + \mathbb{E} \int_0^\tau \left| \sigma(u + t, Z^{t, z}_{t+u}) \right|^2 \, du \right) \nonumber \\
&\leq& C \left( |z|^2 + \mathbb{E} \int_0^\tau (1 + |Z^{t, z}_{t+u}|^2) \, du \right),
\label{eq:A1}
\end{eqnarray}
for any $\tau \in \mathcal{T}_{0,h}$. In particular, for any deterministic time $\tau = s$, Fubini's theorem and Gronwall's lemma imply that
    \begin{equation}
    \mathbb{E} |Z^{t, z}_{t+s}|^2 \leq C (1 + |z|^2), \forall s \in [0,h].
    \label{eq:A.2}
    \end{equation}
    We obtain $(a)$ by injecting \eqref{eq:A.2} into \eqref{eq:A1} and noting that
    \[
    \mathbb{E} \int_0^\tau |Z^{t, z}_{t+u}|^2 \, du \leq \int_0^h \mathbb{E} |Z^{t, z}_{t+u}|^2 \, du
    \]
    for any stopping time $\tau \in \mathcal{T}_{0,h}$.

We prove $(b)$. Similar arguments as above and \eqref{eq:A.2} imply that for all $\tau \in \mathcal{T}_{0,h}$
    \begin{align*}
    \mathbb{E} |Z^{t, z}_{t+\tau} - z|^2 \leq C \int_0^h \left( 1 + \mathbb{E} |Z^{t, z}_{t+u}|^2 \right) \, du \leq C (1 + |z|^2) h.
    \end{align*}

Finally, we consider $(c)$. From \cite{Detemple2017}, we easily check that, for every fixed $s \geq 0$ and $z = (x,y), z' = (x',y') \in \mathbb{R} \times \mathbb{R}^+$, we have
\begin{equation}\label{eq:CIR_fixed_time}
\mathbb{E}\left[ |Y_{t+s}^{t,x,y} - Y_{t+s}^{t,x',y'}| \right] \leq C|y - y'|.
\end{equation}
From this, we have:
\begin{align*}
&\mathbb{E}\Bigl[\sup_{0\le s\le h}|X_{t+s}^{t,x,y}-X_{t+s}^{t,x',y'}|\Bigr]\\
&\qquad\le C \left( |x - x'| + \int_t^{t+h} \mathbb{E}[ |Y_{t+u}^{t,x,y} - Y_{t+u}^{t,x',y'}| ] \, du + \mathbb{E} \left[ \sup_{s \in [t,t+h]} \left| \int_t^s (\sqrt{Y_{t+u}^{t,x,y}} - \sqrt{Y_{t+u}^{t,x',y'}}) \, dB_u \right| \right] \right) \\
&\qquad\le C\Bigl(|x-x'|+\int_t^{t+h}\!\mathbb{E}[|Y_{t+u}^{t,x,y}-Y_{t+u}^{t,x',y'}|]\,du
+\bigl (\mathbb{E}[\sup_{s\in[t,t+h]}|\textstyle\int_t^s(\sqrt{Y_{t+u}^{t,x,y}}-\sqrt{Y_{t+u}^{t,x',y'}})\,dB_u|^2]\bigr)^{1/2}\Bigr) \\
&\qquad\le C\Bigl(|x-x'|+\int_t^{t+h}\!\mathbb{E}[|Y_{t+u}^{t,x,y}-Y_{t+u}^{t,x',y'}|]\,du
+\bigl(\mathbb{E}[\int_t^{t+h}|Y_{t+u}^{t,x,y}-Y_{t+u}^{t,x',y'}|\,du]\bigr)^{1/2}\Bigr) \\
&\qquad\le C\bigl(|x-x'|+|y-y'|+\sqrt{|y-y'|}\bigr).
\end{align*}
It remains to bound the $Y$-component. From the CIR dynamics, we have
\[
Y_{t+s}^{t,x,y} - Y_{t+s}^{t,x',y'} = (y - y') - \kappa \int_0^s (Y_{t+u}^{t,x,y} - Y_{t+u}^{t,x',y'}) \, du + \sigma \int_0^s \left( \sqrt{Y_{t+u}^{t,x,y}} - \sqrt{Y_{t+u}^{t,x',y'}} \right) dW_u.
\]
Taking the supremum and expectation, and applying the Burkholder--Davis--Gundy inequality to the martingale term, we obtain
\begin{align*}
\mathbb{E}\Bigl[\sup_{0 \leq s \leq h} |Y_{t+s}^{t,x,y} - Y_{t+s}^{t,x',y'}|\Bigr]
&\leq |y - y'| + \kappa \int_0^h \mathbb{E}\bigl[|Y_{t+u}^{t,x,y} - Y_{t+u}^{t,x',y'}|\bigr] \, du \\
&\quad + C \left( \mathbb{E} \int_0^h \left| \sqrt{Y_{t+u}^{t,x,y}} - \sqrt{Y_{t+u}^{t,x',y'}} \right|^2 du \right)^{1/2}.
\end{align*}
Using the inequality $|\sqrt{a} - \sqrt{b}|^2 \leq |a - b|$ for $a, b \geq 0$, together with \eqref{eq:CIR_fixed_time}, we deduce
\[
\mathbb{E}\Bigl[\sup_{0 \leq s \leq h} |Y_{t+s}^{t,x,y} - Y_{t+s}^{t,x',y'}|\Bigr]
\leq C\bigl(|y - y'| + \sqrt{|y - y'|}\bigr).
\]
Combining this with the estimate for the $X$-component, we conclude
\[
\mathbb{E}\Bigl[\sup_{0 \leq s \leq h} |Z_{t+s}^{t,z} - Z_{t+s}^{t,z'}|\Bigr]
\leq \mathbb{E}\Bigl[\sup_{0 \leq s \leq h} |X_{t+s}^{t,z} - X_{t+s}^{t,z'}|\Bigr] + \mathbb{E}\Bigl[\sup_{0 \leq s \leq h} |Y_{t+s}^{t,z} - Y_{t+s}^{t,z'}|\Bigr]
\leq C\bigl(|z - z'| + \sqrt{|z - z'|}\bigr).\]
\end{proof}

\begin{lemma}
\label{lem:Lipschitz_cons_1}
There exists a constant \( C > 0 \) such that for all \( (t, z), (t,z') \in [0, T] \times \mathbb{R} \times \mathbb{R}^+ \)
\[
|P^A(t,z) - P^A(t, z')| \leq C (|z - z'| + \sqrt{|z-z'|}).
\]
\end{lemma}

\begin{proof}
Using the definition of \( P^A \) in \eqref{def: PA_american_option}, we estimate
\begin{align*}
|P^A(t, z) - P^A(t, z')|
&\leq \sup_{\tau \in \mathcal{T}_{t,T}} \mathbb{E} \left[ \left| G(X_\tau^{t,z}) - G(X_\tau^{t,z'}) \right| \right] \\
&\leq \mathbb{E} \left[ \sup_{t \leq s \leq T} \left| G( X_s^{t,z}) - G( X_s^{t,z'}) \right| \right] \\
&\leq C \, \mathbb{E} \left[ \sup_{t \leq s \leq T} \left| X_{s}^{t,z} - X_{s}^{t,z'} \right| \right] \\
&\leq C \, \mathbb{E} \left[ \sup_{t \leq s \leq T} \left| Z_{s}^{t,z} - Z_{s}^{t,z'}  \right| \right] \\
&\leq C (|z - z'| + \sqrt{|z-z'|}).
\end{align*}
The first inequality uses \(|\sup A - \sup B| \leq \sup |A - B|\) and the triangle inequality inside \(\mathbb{E}\). The second inequality holds because for any stopping time \(\tau \in \mathcal{T}_{t,T}\), we have \(\tau \in [t, T]\), so \(|G(X_\tau) - G(X'_\tau)| \leq \sup_{t \leq s \leq T} |G(X_s) - G(X'_s)|\) pointwise; taking expectation and then supremum over \(\tau\) yields the result. The third inequality follows from the Lipschitz continuity of \( G(x) = (K - e^{x})^+ \), and the last inequality follows from Lemma~\ref{lem:stochastic_bound_1}~(c).
\end{proof}

\begin{lemma}
\label{lem:time_continous}
For any $z \in \mathbb{R} \times \mathbb{R}^+$ and $0 \leq t < s \leq T$, we have
\[
0 \leq P^{A}(t, z) - P^{A}(s, z) \leq C \left( 1 + |z| \right) |t - s|^{1/2} + C \left( 1 + |z|^{1/2} \right) |t - s|^{1/4},
\]
where $C > 0$ is a constant independent of $t$, $s$, and $z$.
\end{lemma}

\begin{proof}
Let \( 0 \leq t < s \leq T \), with \( h = s - t > 0 \). From Proposition 3.1 of \citet{huyenpham}, we have
\begin{align}
P^{A}(t, x, y) &= \sup_{\tau \in \mathcal{T}_{0,T-t}} \mathbb{E} \left[ 1_{\{\tau < h\}} e^{-r \tau} G( X_{t+\tau}^{t, x,y}) + 1_{\{\tau \geq h\}} e^{-r h}P^{A}(t + h, X_{t+h}^{t, x,y}, Y_{t+h}^{t,x, y}) \right].
\label{eq:CA_t}
\end{align}
We rewrite $
P^{A}(s, x, y) = 1_{\{\tau \geq h\}}P^{A}(s, x, y) + 1_{\{\tau < h\}}P^{A}(s, x, y).$ From \eqref{eq:CA_t}, we estimate
\begin{eqnarray}
0 &\leq& P^{A}(t, x, y) -P^{A}(s, x, y) \notag  \\
&\leq& \sup_{\tau \in \mathcal{T}_{0,T-t}} \mathbb{E} \Bigg[ 1_{\{\tau < h\}} e^{-r \tau} \left( G( X_{t+\tau}^{t, x,y}) - G(x) \right) \notag \\
&& \qquad + 1_{\{\tau \geq h\}} e^{-r h} \left(P^{A}(s, X_s^{t, x,y}, Y_s^{t,x, y}) -P^{A}(s, x, y) \right) \notag  + 1_{\{\tau < h\}} e^{-r \tau} \left( G(x) -P^{A}(s, x, y) \right) \Bigg].
\label{eq:DPP_diff}
\end{eqnarray}
\smallskip
\noindent We now bound each term in the expectation. For the first term, by utilising the fact that \(G\) is Lipschitz continuous, and Lemma~\ref{lem:stochastic_bound_1}~(b), we obtain
\begin{eqnarray}
\mathbb{E} \left[ 1_{\{\tau < h\}} e^{-r \tau} \left( G( X_{t+\tau}^{t,x, y}) - G(x) \right) \right]
&\leq& \mathbb{E} \left| G(X_{t+\tau}^{t,x, y}) - G(x) \right| \notag \\
&\leq& C  \mathbb{E} \left| X_{t+\tau}^{t, x,y} - x \right|  \leq C \left( 1 + |z| \right) h^{1/2}.
\label{eq:term1_1}
\end{eqnarray}
For the second term, by Lemma~\ref{lem:Lipschitz_cons_1} and Lemma~\ref{lem:stochastic_bound_1}~(b), we have
\begin{eqnarray}
\mathbb{E} \left[ 1_{\{\tau \geq h\}} e^{-r h} \left(P^{A}(s, X_s^{t, x, y}, Y_s^{t, x, y}) -P^{A}(s, x, y) \right) \right]
&\leq& \mathbb{E} \left|P^{A}(s, X_s^{t, x,y}, Y_s^{t, x,y}) -P^{A}(s, x, y) \right| \notag \\
&\leq& C  \left(\mathbb{E}\left |Z_{s}^{t,x,y}-z \right| + \mathbb{E}\sqrt{|Z_{s}^{t,x,y}-z|} \right) \notag\\
&\leq& C \left( 1 + |z| \right) h^{1/2} + C \left( 1 + |z|^{1/2} \right) h^{1/4}.
\label{eq:term2_1}
\end{eqnarray}
For the third term, since \( G(x) -P^{A}(s, x, y) \leq 0 \), it is obvious that
\begin{align}
\mathbb{E} \left[ 1_{\{\tau < h\}} e^{-r \tau} \left( G(x) - P^{A}(s, x, y) \right) \right] \leq 0.
\label{eq:term3_1}
\end{align}
Combining the bounds from \eqref{eq:term1_1}, \eqref{eq:term2_1}, \eqref{eq:term3_1}, we conclude
\[
0 \leq P^{A}(t, z) -P^{A}(s, z) \leq C \left( 1 + |z| \right) h^{1/2} + C \left( 1 + |z|^{1/2} \right) h^{1/4}.
\]
The proof is complete.
\end{proof}

\begin{remark}
From Lemma \ref{lem:Lipschitz_cons_1} and Lemma \ref{lem:time_continous}, we deduce that for any \( (t, z_1), (s,z_2) \in [0, T] \times \mathbb{R} \times \mathbb{R}^+ \),
\begin{equation}
\begin{aligned}
    |P^{A}(t, z_1) - P^{A}(s, z_2)|
    &\leq |P^{A}(t, z_1) - P^{A}(s, z_1)| + |P^{A}(s, z_1) - P^{A}(s, z_2)| \\
    &\leq C \left( 1 + |z_1| \right) |t-s|^{1/2} +    C\left( 1 + |z_1|^{1/2} \right) |t-s|^{1/4} +
    C \left(|z_1 - z_2| + \sqrt{|z_1 - z_2|}\right).
\end{aligned}
\end{equation}

Consequently, \(P^A \in C([0, T] \times \mathbb{R} \times \mathbb{R}^+)\).
\end{remark}

Having established the continuity of \(P^A\), we are now in a position to characterise it as a viscosity solution. The continuity is crucial here: it ensures that the definition of viscosity solution applies and that the comparison principle, which underpins uniqueness, can be carried out.

\section{Viscosity solution}\label{sec:viscosity_solution}
In this section, we introduce several essential definitions related to viscosity solutions. Following this, we establish that the American option value function is the unique viscosity solution to Equation~\eqref{eq:heston_pde} in Proposition \ref{lem:viscosity_solution} and
Proposition \ref{lem:unique_viscosity_solution2}.

\begin{definition}
\label{def:viscosity}
A function \( v \in C([0, T] \times \mathbb{R} \times \mathbb{R}^+) \) is called a \emph{viscosity solution} of Equation~\eqref{eq:heston_pde} if it satisfies both the supersolution and subsolution properties outlined below:
\begin{itemize}
\item[(i)] \emph{Viscosity supersolution}: For every test function \( \varphi \in C^{1,2}([0, T] \times \mathbb{R} \times \mathbb{R}^+)\) such that \( v - \varphi \) achieves a global minimum at some point \( (t_0, x_0, y_0) \in [0, T) \times \mathbb{R} \times \mathbb{R}^+ \), the following inequality holds:
\[
\min \left\{ -\partial_t \varphi(t_0, x_0, y_0) - \mathcal{L} \varphi(t_0, x_0, y_0) + r v(t_0, x_0, y_0), \, v(t_0, x_0, y_0) - G(x_0) \right\} \geq 0,
\]
and \( v(T, x, y) \geq G(x) \) for all \( (x, y) \in \mathbb{R} \times \mathbb{R}^+ \).

\item[(ii)] \emph{Viscosity subsolution}: For every test function \( \varphi \in C^{1,2}([0, T] \times \mathbb{R} \times \mathbb{R}^+) \) such that \( v - \varphi \) achieves a global maximum at some point \( (t_0, x_0, y_0) \in [0, T) \times \mathbb{R} \times \mathbb{R}^+ \), the following inequality holds:
\[
\min \left\{ -\partial_t \varphi(t_0, x_0, y_0) - \mathcal{L} \varphi(t_0, x_0, y_0) + r v(t_0, x_0, y_0), \, v(t_0, x_0, y_0) - G(x_0) \right\} \leq 0,
\]
and \( v(T, x, y) \leq G(x) \) for all \( (x, y) \in \mathbb{R} \times \mathbb{R}^+ \).
\end{itemize}
\end{definition}

\begin{definition}
\label{def:jets2}
Let
\[
\mathcal{X} = \begin{pmatrix}
\partial^2_{xx} & \partial^2_{xy} & 0 \\
\partial^2_{yx} & \partial^2_{yy} & 0 \\
0 & 0 & 0
\end{pmatrix}
\]
denote a $3 \times 3$ symmetric matrix with the specified zero entries in the third row and column. Define the following submatrices:
\[
\mathcal{X}_1 = \begin{pmatrix} \partial^2_{xx} \end{pmatrix}, \quad
\mathcal{X}_2 = \begin{pmatrix} \partial^2_{xy} & 0 \end{pmatrix}, \quad
\mathcal{X}_3 = \begin{pmatrix} \partial^2_{yy} & 0 \\ 0 & 0 \end{pmatrix}, \quad
\mathcal{X}_4 = \begin{pmatrix} \partial^2_{xx} & \partial^2_{xy} \\ \partial^2_{yx} & \partial^2_{yy} \end{pmatrix}.
\]
Let $p = (p_1, p_2, p_3) \in \mathbb{R}^3$. For a continuous function $v \in C([0, T] \times \mathbb{R} \times \mathbb{R}^+)$ and a point $(t, z) \in [0, T] \times \mathbb{R} \times \mathbb{R}^+$, the \emph{parabolic superjet} $J^{2,+} v(t, z)$ is the set of all pairs $(p, \mathcal{X}) \in  \mathbb{R}^3 \times \mathbb{S}_3$ such that
\begin{equation*}
v(s, z') \leq v(t, z) + p_3 (s - t) + \begin{pmatrix} p_1 & p_2 \end{pmatrix} (z' - z) + \frac{1}{2} (z' - z)^T \mathcal{X}_4 (z' - z) + o(|s - t| + |z' - z|^2)
\end{equation*}
as $(s, z') \to (t, z)$. The \emph{parabolic subjet} is defined as $J^{2,-} v(t, z) = -J^{2,+} (-v)(t, z)$.
\end{definition}

\begin{proposition}
\label{lem:viscosity_solution}
The value function for an American put option \( P^A(t, x, y) \) is a viscosity solution of Equation \eqref{eq:heston_pde}.
\end{proposition}
\begin{proof}
From Definition~\ref{def:viscosity}, to establish that \(P^A\) is a viscosity solution of Equation \eqref{eq:heston_pde}, it is essential to prove that it serves as both a viscosity supersolution and a viscosity subsolution. The proof is organized into two principal steps.

Step 1: We prove the supersolution property. Consider a point \( (t, x, y) \in [0, T) \times \mathbb{R} \times \mathbb{R}^+ := Q_T \) and a test function \( \varphi \in C^{1,2}([0, T] \times \mathbb{R} \times \mathbb{R}^+) \) such that, without loss of generality,
\begin{equation}
\label{eq:supersolution_setting}
    0 = (P^A - \varphi)(t, x, y) = \min_{(s, \xi, \zeta) \in Q_T} (P^A - \varphi)(s, \xi, \zeta).
\end{equation}
It is clear that $P^A(t,x,y) \ge G(x).$ Let \( (t_n, x_n, y_n), n \in \mathbb{N}\) be a sequence in \( Q_T \) converging to \( (t, x, y) \); then
\[
P^A(t_n, x_n, y_n) \to P^A(t, x, y)
\]
as $n$ tends to infinity. We define
\begin{equation}
\label{eq:eta}
\eta_n = P^A(t_n, x_n, y_n) - \varphi(t_n, x_n, y_n) > 0,
\end{equation}
which approaches zero as \( n \to \infty \). For each \( n \in \mathbb{N} \), we introduce the stopping time
\begin{equation}
    \label{eq:bound_unity_ball}
\theta_n := \inf \left\{ s > t_n : \left( s - t_n, X_{s}^{t_n,x_n,y_n} - x_n, Y_s^{t_n,x_n,y_n} - y_n \right) \notin [0, h_n) \times \mathcal{B} \times \mathcal{B} \right\},
\end{equation}
where \( \mathcal{B} \) denotes the unit ball in \( \mathbb{R} \), and
\[
h_n := \sqrt{\eta_n} \, 1_{\{\eta_n \neq 0\}} + \frac{1}{n} \, 1_{\{\eta_n = 0\}}.
\]
As \( n \to \infty \), we have that $h_n \to 0$, and thus \( \theta_n \to t \). By the supermartingale property of American option prices, we have
\begin{equation}
\label{eq: American_martigel}
    P^A(t_n, x_n, y_n) \geq \mathbb{E} \left[ e^{-r(h - t_n)} P^A\left( h, X_h^{t_n, x_n,y_n}, Y_h^{t_n, x_n,y_n} \right) \right],
\end{equation}
for any \( h \geq t_n \). Applying Itô's formula to \( e^{-r(h - t_n)} \varphi(h, X_h^{t_n, x_n,y_n}, Y_h^{t_n, x_n,y_n}) \) and using the boundedness from \eqref{eq:bound_unity_ball}, we obtain
\begin{align*}
\mathbb{E} \left[ e^{-r(\theta_n - t_n)} \varphi \left( \theta_n, X_{\theta_n}^{t_n, x_n,y_n}, Y_{\theta_n}^{t_n, x_n,y_n} \right) \right]
&= \varphi(t_n, x_n, y_n) \\
&\quad + \mathbb{E} \left[ \int_{t_n}^{\theta_n} e^{-r(s - t_n)} \left( \partial_t + \mathcal{L} - r \right) \varphi \left( s, X_s^{t_n, x_n,y_n}, Y_s^{t_n, x_n,y_n} \right) \, ds \right].
\end{align*}
According to \eqref{eq:supersolution_setting}, it follows that $\varphi(\theta_n, X_{\theta_n}^{t_n, x_n, y_n}, Y_{\theta_n}^{t_n, x_n, y_n}) \leq P^A(\theta_n, X_{\theta_n}^{t_n, x_n, y_n}, Y_{\theta_n}^{t_n, x_n, y_n})$. Furthermore, by \eqref{eq:eta}, we obtain
\begin{align*}
\mathbb{E} \left[ e^{-r(\theta_n - t_n)} P^A \left( \theta_n, X_{\theta_n}^{t_n, x_n,y_n}, Y_{\theta_n}^{t_n, x_n,y_n} \right) \right]
&\geq P^A(t_n, x_n, y_n) - \eta_n \\
&\quad + \mathbb{E} \left[ \int_{t_n}^{\theta_n} e^{-r(s - t_n)} \left( \partial_t + \mathcal{L} - r \right) \varphi \left( s, X_s^{t_n, x_n,y_n}, Y_s^{t_n, x_n,y_n} \right) \, ds \right].
\end{align*}
By \eqref{eq: American_martigel}, we obtain
\[
0 \geq -\eta_n + \mathbb{E} \left[ \int_{t_n}^{\theta_n} e^{-r(s - t_n)}\left( \partial_t + \mathcal{L} - r \right) \varphi \left( s, X_s^{t_n, x_n,y_n}, Y_s^{t_n, x_n,y_n} \right) \, ds \right],
\]
and thus,
\[
\mathbb{E} \left[ \int_{t_n}^{\theta_n} \frac{1}{h_n} e^{-r(s - t_n)}\left( \partial_t + \mathcal{L} - r \right) \varphi \left( s, X_s^{t_n, x_n,y_n}, Y_s^{t_n, x_n,y_n} \right) \, ds \right] \leq \frac{\eta_n}{h_n}.
\]
Note that for $n \geq N(\omega)$ sufficiently large, $\theta_n(\omega) = t_n+h_n $. Since the sequence of random variables
\[
\int_{t_n}^{\theta_n} \frac{1}{h_n} e^{-r(s - t_n)}\left( \partial_t + \mathcal{L} - r \right) \varphi \left( s, X_s^{t_n, x_n,y_n}, Y_s^{t_n, x_n,y_n} \right) \, ds
\]
is uniformly bounded in \( n \), on the interval \([t_n,\theta_n]\). Applying the dominated convergence theorem and the mean value theorem, we obtain
\[
\left( \partial_t + \mathcal{L} - r \right) \varphi(t, x, y) \leq 0 \quad \text{or equivalently,} \quad -\left( \partial_t + \mathcal{L} - r \right) \varphi(t, x, y) \geq 0.
\]
Thus, \( P^A(t, x, y) \) is a viscosity supersolution.

Step 2: We prove the subsolution property. Let \( (t, x, y) \in Q_T \) and \( \varphi \in C^{1,2}([0, T] \times \mathbb{R} \times \mathbb{R}^+) \) such that, without loss of generality,
\begin{equation}
\label{eq:super_solution_1}
0 = (P^A - \varphi)(t, x, y) = \max_{(s, \xi, \zeta) \in Q_T} (P^A - \varphi)(s, \xi, \zeta).
\end{equation}
We need to prove that
\begin{equation}\label{eq:sub}
\min \Bigl\{ (-\partial_t - \mathcal{L} + r) \varphi(t,x,y),\
             P^A(t,x,y) - G(x) \Bigr\} \leq 0.
\end{equation}
From the fact that \( P^A(t, x, y) \geq G(x) \), if \( P^A(t, x, y) = G(x) \) then \eqref{eq:sub} holds trivially. Therefore, we consider the case \( P^A(t, x, y) > G(x) \) and also \(\left( -\partial_t - \mathcal{L} + r \right) \varphi(t,x,y) > 0\).
We introduce the stopping time
\[
\tau_{0} := \inf \left\{ 0 \leq s \leq T - t : P^A \left( s + t, X_{t+s}^{t, x,y}, Y_{t+s}^{t,x, y} \right) = G \left( X_{t+s}^{t, x,y} \right)  \right\} \wedge (T-t),
\]
which is also the optimal stopping time. Since
\(\left( -\partial_t - \mathcal{L} + r \right) \varphi\) is continuous at \((t,x,y)\), there exists \(0 < \delta < T-t\) such that
\[
\left( -\partial_t - \mathcal{L} + r \right) \varphi(t', x', y')  > \delta
\quad \text{for all } (t', x', y') \in B_{\delta}(t, x, y),
\]
where $B_\delta(t,x,y) = \{(t',x',y') \in Q_T : (t'-t)^2 + (x'-x)^2 + (y'-y)^2 \le \delta^2 \}$. Now let \(\tau_1\) be the stopping time such that
\begin{equation}
\tau_1 := \inf \left\{ 0 \leq s < T - t : \left( t+s, X_{t+s}^{t,x,y}, Y_{t+s}^{t,x,y} \right) \notin B_{\delta}(t, x, y)\right\}.
\end{equation}
Then let \(\tau = \tau_0 \wedge \tau_1\). From our assumptions, and the definitions of \(\tau_0\), \(\tau_1\), we have \(\tau_0 > 0\) and \(\tau_1 > 0\) a.s. Since \(\tau \le \tau_0\) and by the martingale property of the American option up to \(\tau_0\), we obtain
\[
P^A(t, x, y) = \mathbb{E} \left[e^{-r\tau} P^A \left( t + \tau, X_{t + \tau}^{t, x,y}, Y_{t + \tau}^{t,x, y} \right) \, \right].
\]
Moreover, for every \( s \in [0, \tau] \), the point \( (t + s, X_{t+s}^{t, x, y}, Y_{t+s}^{t, x, y}) \) belongs to \( B_\delta(t, x, y) \). Consequently, we have
\begin{align*}
0 &= \mathbb{E} \left[e^{-r\tau} P^A \left( t + \tau, X_{t + \tau}^{t, x,y}, Y_{t + \tau}^{t,x, y} \right) \, \right] - P^A(t, x, y) \\
&\leq \mathbb{E} \left[e^{-r\tau} \varphi\left( t + \tau, X_{t + \tau}^{t, x,y}, Y_{t + \tau}^{t,x, y} \right) \, \right] - \varphi(t, x, y) \\
&= -\mathbb{E} \int_0^\tau e^{-rs} \left( -\partial_t - \mathcal{L} + r \right)\varphi(t + s, X_{t+s}^{t, x,y}, Y_{t+s}^{t,x, y})\, ds \\
&\leq -\delta\, e^{-rT}\, \mathbb{E}\tau,
\end{align*}
where the first inequality comes from \eqref{eq:super_solution_1}, and the last inequality uses \(e^{-rs} \geq e^{-rT}\).

This implies that $\tau = 0$ a.s. On the other hand, $\tau > 0$ a.s.\ by definition. This is a contradiction, which proves that \(\left( -\partial_t - \mathcal{L} + r \right) \varphi(t,x,y) \le 0\); that is, $P^A$ is a viscosity subsolution of Equation~\eqref{eq:heston_pde}.

Finally, since $P^A(T, x, y) = G(x)$ by definition \eqref{def: PA_american_option}, the terminal conditions $P^A(T, x, y) \geq G(x)$ and $P^A(T, x, y) \leq G(x)$ required by Definition~\ref{def:viscosity} are both satisfied.
\end{proof}

We prove the uniqueness of viscosity solutions by adapting the arguments of \cite{ishii2021existence}.

\begin{proposition}
\label{lem: u_v_sub_super}
Recall \( Q_T \), which was mentioned in the proof of Proposition \ref{lem:viscosity_solution}.
Suppose \(u\) is a viscosity subsolution and \(v\) is a viscosity supersolution of Equation~\eqref{eq:heston_pde} satisfying the conditions
\begin{itemize}
\item \(\max_{Q_T} |u| \le K\) and \(\max_{Q_T}|v| \le K\),
\item \(u, v \in C([0,T] \times \mathbb{R} \times \mathbb{R}^{+})\),
\item \(u(T, \cdot) \leq v(T, \cdot)\) for all \((x, y) \in \mathbb{R} \times \mathbb{R}^+\).
\end{itemize}
Then, it holds that
\begin{equation}
u(t, \cdot) \leq v(t, \cdot) \quad \text{in } Q_T.
\label{eq:52}
\end{equation}
\end{proposition}
\begin{proof}
In this proof, we define the operator \(\mathcal{M}u = -\partial_t u - \mathcal{L}u\). Additionally, we introduce the function \(F: \mathbb{R}^+ \times \mathbb{R}^3 \times \mathbb{S}_3 \to \mathbb{R}\) as follows:
\[
F(y, p, \mathcal{X}) = -p_3 - \frac{y}{2} \mathcal{X}_1 - \langle \rho \sigma y \mathcal{X}_2, \textbf{1} \rangle - \operatorname{Tr}\left( \frac{1}{2} \sigma^2 y \mathcal{X}_3 \right) - \left( r -\delta - \frac{y}{2} \right) p_1 - \kappa (\theta - y) p_2,
\]
where \(\textbf{1}=(1,1)^{T}\) and \(p = (p_1, p_2, p_3) \in \mathbb{R}^3\).
With these notations, we have
\[
\mathcal{M}\phi(x, y, t) = F(y, D\phi(x, y, t), D^2\phi(x, y, t)).
\]
We divide our proof into four steps.

\textbf{Step 1:} We define the functions \(\rho: \mathbb{R} \times \mathbb{R}^+ \to \mathbb{R}\) and \(f: Q_T \to \mathbb{R}\) as follows:
\[
\rho(x,y) = 1 + x^2 + y^2, \quad f(t,x,y) = \rho(x,y) e^{-Ct},
\]
where \(C\) is a constant that will be specified subsequently. A straightforward calculation reveals that \(\partial_t f = -C f\) on \(Q_T\). It is easily verified that there exists a constant \(C_1 > 0\) such that
\[
\mathcal{L}f \leq C_1 f \quad \text{on } Q_T.
\]
Notably, the constant \(C_1\) is independent of both \(C\) and \(t\). Consequently, we obtain
\[
(\mathcal{M} + r)f \geq (C - C_1 + r)f \quad \text{on } Q_T.
\]
By selecting \(C = C_1 - r\), we ensure that
\begin{equation}
(\mathcal{M} + r)f \geq 0 \quad \text{on } Q_T. \label{eq:mf-ineq}
\end{equation}

For any \(\varepsilon > 0\), we introduce the function \(u_\varepsilon\) on \(\overline{Q}_T\) defined by
\[
u_\varepsilon(t,x,y) = u(t,x,y) - \varepsilon f(t,x,y).
\]
Given that \(f\) is a classical solution of \eqref{eq:mf-ineq},
and \(f > 0\) on \(Q_T\), it follows that \(u_\varepsilon\) is a viscosity subsolution to the equation
\[
\min\{(\mathcal{M} + r)u_\varepsilon, u_\varepsilon - G\} = 0 \quad \text{on } Q_T.
\]
It is clear that
\begin{equation}\label{eq:53}
u_\varepsilon(t,x,y) \le K - \varepsilon e^{-CT}(1+x^2+y^2).
\end{equation}

\textbf{Step 2:}
Given that \(u_\varepsilon \leq u\), it immediately follows that
\[
\sup_{\overline{Q}_T} (u_\varepsilon - v) \leq \sup_{\overline{Q}_T} (u - v).
\]
It is straightforward to verify that if \(\sup_{\overline{Q}_T} (u_\varepsilon - v) \leq 0\) holds for all \(\varepsilon > 0\), then \(\sup_{\overline{Q}_T} (u - v) \leq 0\). Thus, to establish \eqref{eq:52}, it is sufficient to show that
\begin{equation}
u_\varepsilon \leq v \quad \text{on } \overline{Q}_T
\label{eq:dpcm}
\end{equation}
for every \(\varepsilon > 0\). We proceed by contradiction. Suppose there exists \(\varepsilon_0 > 0\) such that
\begin{equation}
\sup_{(t,x,y) \in \overline{Q}_T} \bigl( u_{\varepsilon_0}(t,x,y) - v(t,x,y) \bigr) > 0.
\label{eq:54}
\end{equation}
For \(\gamma > 0\), introduce the function \(u_{\varepsilon_0}^\gamma(t,x,y) := u_{\varepsilon_0}(t,x,y) + \gamma(t - T)\). This satisfies \(u_{\varepsilon_0}^\gamma \leq u_{\varepsilon_0}\) on \(\overline{Q}_T\) and is a subsolution to
\[
\min \bigl\{ (\mathcal{M} + r) u_{\varepsilon_0}^\gamma(t,x,y) + \gamma, \, u_{\varepsilon_0}^\gamma(t,x,y) - G \bigr\} = 0 \quad \text{in } Q_T.
\]
To confirm \eqref{eq:dpcm}, it suffices to show that \(u_{\varepsilon_0}^\gamma(t,x,y) \leq v(t,x,y)\) for all \((t,x,y) \in \overline{Q}_T\) and every \(\gamma > 0\). Therefore, we may assume, by substituting \(u_{\varepsilon_0}\) with \(u_{\varepsilon_0}^\gamma\) when necessary, that \(u_{\varepsilon_0}\) is a subsolution of
\[
\min \bigl\{ (\mathcal{M} + r) u + \gamma, \, u - G \bigr\} = 0 \quad \text{in } Q_T.
\]
Let \(\alpha > 0\) and \(\beta > 0\), and consider the auxiliary function
\[
\Phi(t,x,y,\tau,\xi,\eta) := u_{\varepsilon_0}(t,x,y) - v(\tau,\xi,\eta) - \alpha(x - \xi)^2 - \beta(y - \eta)^2 - \beta(t - \tau)^2
\]
defined on \(\overline{Q}_T \times \overline{Q}_T\). From \eqref{eq:53}, combined with the fact that \( |v| \leq K\) and \(u, v \in C(\overline{Q}_T)\), the function \(\Phi\) attains a global maximum. Denote a point of maximum by \((t_{\alpha\beta}, x_{\alpha\beta}, y_{\alpha\beta}, \tau_{\alpha\beta}, \xi_{\alpha\beta}, \eta_{\alpha\beta})\). From \eqref{eq:53}, it follows that as \((\alpha, \beta) \to (\infty, \infty)\), the sequence of maximum points remains bounded. Observing that
\begin{equation}
\label{eq:max_ine_1}
\begin{aligned}
\max_{\overline{Q}_T} (u_{\varepsilon_0} - v) &= \max_{(t,x,y) \in \overline{Q}_T} \Phi(t,x,y,t,x,y) \\
&\leq \Phi(t_{\alpha\beta}, x_{\alpha\beta}, y_{\alpha\beta}, \tau_{\alpha\beta}, \xi_{\alpha\beta}, \eta_{\alpha\beta}) \\
&\leq u_{\varepsilon_0}(t_{\alpha\beta}, x_{\alpha\beta}, y_{\alpha\beta}) - v(\tau_{\alpha\beta}, \xi_{\alpha\beta}, \eta_{\alpha\beta}),
\end{aligned}
\end{equation}
and from \eqref{eq:54}, we deduce that
\begin{equation}
\sup_{\alpha > 1, \, \beta > 1} \Bigl( \alpha(x_{\alpha\beta} - \xi_{\alpha\beta})^2 + \beta(y_{\alpha\beta} - \eta_{\alpha\beta})^2 + \beta(t_{\alpha\beta} - \tau_{\alpha\beta})^2 \Bigr) < \infty.
\label{eq:sup-bound}
\end{equation}
Moreover, given the boundedness of \((t_{\alpha\beta}, x_{\alpha\beta}, y_{\alpha\beta}, \tau_{\alpha\beta}, \xi_{\alpha\beta}, \eta_{\alpha\beta})\) and \eqref{eq:sup-bound}, for any sequences \(\{\alpha_k\}\) and \(\{\beta_k\}\) with \(\alpha_k \to \infty\) and \(\beta_k \to \infty\), there exists a subsequence such that, as \((\alpha, \beta) \to (\infty, \infty)\),
\[
(t_{\alpha\beta}, x_{\alpha\beta}, y_{\alpha\beta}, \tau_{\alpha\beta}, \xi_{\alpha\beta}, \eta_{\alpha\beta}) \to (t_0, x_0, y_0, t_0, x_0, y_0).
\]
From \eqref{eq:max_ine_1}, we also observe that
\begin{align*}
\max_{\overline{Q}_T} (u_{\varepsilon_0} - v) &\leq \lim_{(\alpha,\beta) \to (\infty,\infty)} \Phi(t_{\alpha\beta}, x_{\alpha\beta}, y_{\alpha\beta}, \tau_{\alpha\beta}, \xi_{\alpha\beta}, \eta_{\alpha\beta}) \\
&\leq u_{\varepsilon_0}(t_0, x_0, y_0) - v(t_0, x_0, y_0) \leq \max_{\overline{Q}_T} (u_{\varepsilon_0} - v).
\end{align*}
Consequently, \((t_0, x_0, y_0)\) is a maximum point of \(u_{\varepsilon_0} - v\), and from \eqref{eq:54}, this yields \(t_0 \neq T\). We can thus find the sequences \(\{\alpha_k\}_{k \in \mathbb{N}}\) and \(\{\beta_m\}_{m \in \mathbb{N}}\) such that
\begin{equation}
\begin{cases}
\displaystyle \lim_{k \to \infty} \lim_{m \to \infty} (t_{\alpha_k \beta_m}, x_{\alpha_k \beta_m}, y_{\alpha_k \beta_m}, \tau_{\alpha_k \beta_m}, \xi_{\alpha_k \beta_m}, \eta_{\alpha_k \beta_m}) = (t_0, x_0, y_0, t_0, x_0, y_0), \\[12pt]
\displaystyle \lim_{k \to \infty} \lim_{m \to \infty} \Bigl[ \alpha_k (x_{\alpha_k \beta_m} - \xi_{\alpha_k \beta_m})^2 + \beta_m (y_{\alpha_k \beta_m} - \eta_{\alpha_k \beta_m})^2 + \beta_m (t_{\alpha_k \beta_m} - \tau_{\alpha_k \beta_m})^2 \Bigr] = 0, \\[12pt]
\displaystyle \lim_{k \to \infty} \lim_{m \to \infty} u_{\varepsilon_0}(t_{\alpha_k \beta_m}, x_{\alpha_k \beta_m}, y_{\alpha_k \beta_m}) = u_{\varepsilon_0}(t_0, x_0, y_0), \\[8pt]
\displaystyle \lim_{k \to \infty} \lim_{m \to \infty} v(\tau_{\alpha_k \beta_m}, \xi_{\alpha_k \beta_m}, \eta_{\alpha_k \beta_m}) = v(t_0, x_0, y_0).
\end{cases}
\label{eq:55}
\end{equation}
Additionally, since \(t_0 \neq T\), we may assume that for all \(\alpha_k\) and \(\beta_m\),
\begin{equation}
(t_{\alpha_k \beta_m}, x_{\alpha_k \beta_m}, y_{\alpha_k \beta_m}) \in Q_T, \quad (\tau_{\alpha_k \beta_m}, \xi_{\alpha_k \beta_m}, \eta_{\alpha_k \beta_m}) \in Q_T.
\label{eq:56}
\end{equation}
Furthermore, given that \(u_{\varepsilon_0}(t_0, x_0, y_0) > v(t_0, x_0, y_0)\) by \eqref{eq:54}, and in view of \eqref{eq:56}, we may assume that for all \(\alpha_k\) and \(\beta_m\),
\begin{equation}
\bigl( u_{\varepsilon_0} - G \bigr)(t_{\alpha_k \beta_m}, x_{\alpha_k \beta_m}, y_{\alpha_k \beta_m}) > \bigl( v - G \bigr)(\tau_{\alpha_k \beta_m}, \xi_{\alpha_k \beta_m}, \eta_{\alpha_k \beta_m}).
\label{eq:57}
\end{equation}
For fixed \(k, m \in \mathbb{N}\), we simplify notations by setting \(\alpha = \alpha_k\) and \(\beta = \beta_m\), and prepare to apply Corollary 4.4 from \cite{ishii2021existence} to the pair \(u_{\varepsilon_0}\) and \(-v\). To this end, we specify
\[
\begin{cases}
n = 3, & n_1 = 1, \quad n_2 = 2, \\[4pt]
U = V = Q_T, & U_1 = V_1 = \mathbb{R}, \quad U_2 = V_2 = \mathbb{R}^+ \times (0,T), \\[4pt]
\hat{\theta} = (t_{\alpha\beta}, x_{\alpha\beta}, y_{\alpha\beta}), & \hat{\zeta} = (\tau_{\alpha\beta}, \xi_{\alpha\beta}, \eta_{\alpha\beta}).
\end{cases}
\]
 We define the smooth functions \(\varphi, \tilde{\varphi}_i \in C^2(U \times V)\) and \(\tilde{\varphi}_i \in C^2(U_i \times V_i)\) for \(i=1,2\) as
\[
\begin{cases}
\varphi(t,x,y,\tau,\xi,\eta) = \alpha(x - \xi)^2 + \beta(y - \eta)^2 + \beta(t - \tau)^2, \\[4pt]
\tilde{\varphi}_1(t,x,y,\tau,\xi,\eta) = \varphi_1(t,x,y,\tau,\xi,\eta) = \alpha(x - \xi)^2, \\[4pt]
\tilde{\varphi}_2(y,t,\eta,\tau) = \varphi_2(y,t,\eta,\tau) = \beta(y - \eta)^2 + \beta(t - \tau)^2,
\end{cases}
\]
and set
\[
\mathcal{A}_1 = D^2 \tilde{\varphi}_1(\hat{\theta}, \hat{\zeta}), \quad \mathcal{A}_2 = D^2 \tilde{\varphi}_2(\hat{\theta}, \hat{\zeta}), \quad \mathcal{A} = D^2 \varphi(\hat{\theta}, \hat{\zeta}).
\]
Direct computation yields

\begin{equation}
\begin{gathered}
\begin{aligned}
\mathcal{A}_1 &= 2\alpha
\begin{pmatrix}
1 & 0 & 0 & -1 & 0 & 0 \\
0 & 0 & 0 & 0 & 0 & 0 \\
0 & 0 & 0 & 0 & 0 & 0 \\
-1 & 0 & 0 & 1 & 0 & 0 \\
0 & 0 & 0 & 0 & 0 & 0 \\
0 & 0 & 0 & 0 & 0 & 0
\end{pmatrix},
&
\mathcal{A}_2 &= 2\beta
\begin{pmatrix}
0 & 0 & 0 & 0 & 0 & 0 \\
0 & 1 & 0 & 0 & -1 & 0 \\
0 & 0 & 1 & 0 & 0 & -1 \\
0 & 0 & 0 & 0 & 0 & 0 \\
0 & -1 & 0 & 0 & 1 & 0 \\
0 & 0 & -1 & 0 & 0 & 1
\end{pmatrix},
\end{aligned} \\[8pt]
\mathcal{A} = \mathcal{A}_1 + \mathcal{A}_2, \quad
\mathcal{A}^2 = \mathcal{A}_1^2 + \mathcal{A}_2^2 = 4\alpha \mathcal{A}_1 + 4\beta \mathcal{A}_2, \quad
|\mathcal{A}_1| = 4\alpha, \quad |\mathcal{A}_2| = 4\beta.
\end{gathered}
\label{eq:58}
\end{equation}

We choose
\[
\varepsilon_1 = \frac{1}{2\alpha}, \quad \varepsilon_2 = \frac{1}{2\beta}, \quad \lambda_i = \frac{1}{\varepsilon_i} + |\mathcal{A}_i| \quad \text{for } i=1,2,
\]
and note that
\[
\lambda_1 = 6\alpha, \quad \lambda_2 = 6\beta.
\]
Let \(\varepsilon = (\varepsilon_1, \varepsilon_2)\) and \(\lambda = (\lambda_1, \lambda_2)\). Noting that
\begin{equation}
\mathcal{A}_\varepsilon = \mathcal{A} + \varepsilon_1 \cdot 4\alpha \mathcal{A}_1 + \varepsilon_2 \cdot 4\beta \mathcal{A}_2 = 3\mathcal{A},
\label{eq:59}
\end{equation}
we define \(\varphi_\varepsilon\) on \(Q_T \times Q_T\) by
\[
\varphi_\varepsilon(\theta, \zeta) = \varphi(\theta, \zeta) + 2 \varphi(\theta - \hat{\theta}, \zeta - \hat{\zeta}),
\]
and observe that
\[
D^2 \varphi_\varepsilon(\hat{\theta}, \hat{\zeta}) = 3\mathcal{A} = \mathcal{A}_\varepsilon.
\]
Furthermore, condition (4.3) of Theorem 4.2 in \cite{ishii2021existence} is satisfied.

\textbf{Step 3:} We define \(w(\theta, \zeta) = u_{\varepsilon_0}(\theta) - v(\zeta)\). Select a compact neighbourhood \(W\) of \((\hat{\theta}, \hat{\zeta})\) within \(Q_T \times Q_T\), and apply Corollary 4.4 from \cite{ishii2021existence} to obtain sequences \(\{(\theta_j, \zeta_j) \subset Q_T \times Q_T\}\), \(\{(\mathcal{X}_j, \mathcal{Y}_j) \subset \mathbb{S}_3 \times \mathbb{S}_3\}\), and \(\{\varphi_j\} \subset C^2(Q_T \times Q_T)\) satisfying condition (4.4) in Theorem 4.2 of \cite{ishii2021existence}, with \(-v\) replacing \(v\). For each \(j \in \mathbb{N}\), condition (4.4) in Theorem 4.2 of \cite{ishii2021existence} implies
\[
\max_{Q_T \times Q_T} (w - \varphi_j) = (w - \varphi_j)(\theta_j, \zeta_j),
\qquad
(D \varphi_j(\theta_j, \zeta_j), \mathcal{X}_j) \in J^{2,+} u_{\varepsilon_0}(\theta_j),
\]
\[
-(D \varphi_j(\theta_j, \zeta_j), \mathcal{Y}_j) \in J^{2,-} v(\zeta_j).
\]
Denote
\[
\theta_j = (t_j, x_j, y_j) \in (0, T) \times \mathbb{R} \times \mathbb{R}^+,
\qquad
\zeta_j = (\tau_j, \xi_j, \eta_j) \in (0, T) \times \mathbb{R} \times \mathbb{R}^+,
\]
and
\[
D_\theta \varphi_j(\theta_j, \zeta_j) = p_j,
\quad
D_\zeta \varphi_j(\theta_j, \zeta_j) = q_j.
\]
This yields
\begin{equation}
\begin{cases}
\min \bigl\{ F(y_j, p_j, \mathcal{X}_j) + r u_{\varepsilon_0}(t_j, x_j, y_j) + \gamma, \, (u_{\varepsilon_0} - G)(t_j, x_j, y_j) \bigr\} \leq 0, \\[8pt]
\min \bigl\{ F(\eta_j, -q_j, -\mathcal{Y}_j) + r v(\tau_j, \xi_j, \eta_j), \, (v - G)(\tau_j, \xi_j, \eta_j) \bigr\} \geq 0.
\end{cases}
\label{ref:510}
\end{equation}
Additionally, from condition (4.4) in Theorem 4.2 of \cite{ishii2021existence}, we have
\begin{equation}
\lim_{j \to \infty} (\theta_j, \zeta_j) = (\hat{\theta}, \hat{\zeta}),
\qquad
\lim_{j \to \infty} \varphi_j = \varphi_\varepsilon \quad \text{in } C^2(U \times V).
\label{511}
\end{equation}
It follows that
\begin{equation}
\begin{cases}
(w - \varphi_\varepsilon)(\hat{\theta}, \hat{\zeta}) = \max_{Q_T \times Q_T} (w - \varphi_\varepsilon) = \lim_{j \to \infty} (w - \varphi_j)(\theta_j, \zeta_j), \\[10pt]
u_{\varepsilon_0}(\hat{\theta}) = \lim_{j \to \infty} u_{\varepsilon_0}(\theta_j),
\qquad
v(\hat{\zeta}) = \lim_{j \to \infty} v(\zeta_j).
\end{cases}
\label{eq:512}
\end{equation}
These relations, combined with \eqref{eq:57}, permit us to assume, after relabelling \(\theta_j\), \(\zeta_j\), and extracting subsequences if needed, that for all \(j \in \mathbb{N}\),
\[
(u_{\varepsilon_0} - G)(t_j, x_j, y_j) > (v - G)(\tau_j, \xi_j, \eta_j).
\]
Together with \eqref{ref:510}, this implies
\begin{equation}
F(y_j, p_j, \mathcal{X}_j) + r u_{\varepsilon_0}(t_j, x_j, y_j) + \gamma \leq 0 \leq F(\eta_j, -q_j, -\mathcal{Y}_j) + r v(\tau_j, \xi_j, \eta_j) \quad \text{for all } j \in \mathbb{N}.
\label{ref:513}
\end{equation}
Furthermore, the inequality
\[
-E_\lambda \leq
\begin{pmatrix}
\mathcal{X}_j & 0 \\
0 & \mathcal{Y}_j
\end{pmatrix}
\leq D^2 \varphi_j(\theta_j, \zeta_j),
\]
where \(E_\lambda \in \mathbb{S}_6\) is the symmetric, positive semidefinite block-diagonal matrix
\[
E_\lambda \;=\;
\begin{pmatrix}
\lambda_1 I_1 & 0 & 0 & 0 \\
0 & \lambda_2 I_2 & 0 & 0 \\
0 & 0 & \lambda_1 I_1 & 0 \\
0 & 0 & 0 & \lambda_2 I_2
\end{pmatrix}
\;=\;
\begin{pmatrix}
6\alpha\, I_1 & 0 & 0 & 0 \\
0 & 6\beta\, I_2 & 0 & 0 \\
0 & 0 & 6\alpha\, I_1 & 0 \\
0 & 0 & 0 & 6\beta\, I_2
\end{pmatrix}.
\]
Then we combine with the convergence
\[
\lim_{j \to \infty} (\theta_j, \zeta_j) = (\hat{\theta}, \hat{\zeta}),
\qquad
\lim_{j \to \infty} \varphi_j = \varphi_\varepsilon \quad \text{in } C^2(Q_T \times Q_T),
\]
ensures that \(\{\mathcal{X}_j, \mathcal{Y}_j\}\) is bounded in \(\mathbb{S}_3 \times \mathbb{S}_3\). Thus, by extracting a subsequence if necessary, we have for some \((\mathcal{X}_{\alpha\beta}, \mathcal{Y}_{\alpha\beta}) \in \mathbb{S}_3 \times \mathbb{S}_3\),
\[
\lim_{j \to \infty} (\mathcal{X}_j, \mathcal{Y}_j) = (\mathcal{X}_{\alpha\beta}, \mathcal{Y}_{\alpha\beta}).
\]
The matrix inequality then gives
\begin{equation}
-E_\lambda \leq
\begin{pmatrix}
\mathcal{X}_{\alpha\beta} & 0 \\
0 & \mathcal{Y}_{\alpha\beta}
\end{pmatrix}
\leq D^2 \varphi_\varepsilon(\hat{\theta}, \hat{\zeta}) = \mathcal{A}_\varepsilon.
\label{ref:514}
\end{equation}
Taking the limit as \(j \to \infty\) in \eqref{ref:513}, we obtain
\begin{equation}
F(y_{\alpha\beta}, p_{\alpha\beta}, \mathcal{X}_{\alpha\beta}) + r u_{\varepsilon_0}(t_{\alpha\beta}, x_{\alpha\beta}, y_{\alpha\beta}) + \gamma \leq 0 \leq F(\eta_{\alpha\beta}, -q_{\alpha\beta}, -\mathcal{Y}_{\alpha\beta}) + r v(\tau_{\alpha\beta}, \xi_{\alpha\beta}, \eta_{\alpha\beta}),
\label{ref:515}
\end{equation}
where \(p_{\alpha\beta} := D_\theta \varphi(t_{\alpha\beta}, x_{\alpha\beta}, y_{\alpha\beta}, \tau_{\alpha\beta}, \xi_{\alpha\beta}, \eta_{\alpha\beta})\) and \(q_{\alpha\beta} := D_\zeta \varphi(t_{\alpha\beta}, x_{\alpha\beta}, y_{\alpha\beta}, \tau_{\alpha\beta}, \xi_{\alpha\beta}, \eta_{\alpha\beta})\).

\textbf{Step 4:} We aim to show that \eqref{ref:515}, together with \eqref{ref:514} and \eqref{eq:55}, leads to the required contradiction. To this end, conforming to the partition \(n = n_1 + n_2 = 1 + 2\) fixed in Step~2, we decompose \(\mathcal{X}_{\alpha\beta}, \mathcal{Y}_{\alpha\beta} \in \mathbb{S}_3\) into blocks as
\[
\mathcal{X}_{\alpha\beta} =
\begin{pmatrix}
X_{1,\alpha\beta} & X_{2,\alpha\beta} \\
X_{2,\alpha\beta}^{\top} & X_{3,\alpha\beta}
\end{pmatrix},
\qquad
\mathcal{Y}_{\alpha\beta} =
\begin{pmatrix}
Y_{1,\alpha\beta} & Y_{2,\alpha\beta} \\
Y_{2,\alpha\beta}^{\top} & Y_{3,\alpha\beta}
\end{pmatrix},
\]
where \(X_{1,\alpha\beta}, Y_{1,\alpha\beta} \in \mathbb{R}\) are the leading (top-left) \(1\times 1\) diagonal blocks, \(X_{3,\alpha\beta}, Y_{3,\alpha\beta} \in \mathbb{S}_2\) are the trailing (bottom-right) \(2\times 2\) diagonal blocks, and \(X_{2,\alpha\beta}, Y_{2,\alpha\beta} \in \mathbb{R}^{1\times 2}\) are the corresponding off-diagonal blocks.

By evaluating the quadratic forms linked to the matrices in \eqref{ref:514} at \((\xi, 0, \eta, 0) \in \mathbb{R}^6\), where \(\xi, \eta \in \mathbb{R}\), and utilizing \eqref{ref:514}, \eqref{eq:58}, and \eqref{eq:59}, we obtain
\[
-6\alpha I_2 \leq
\begin{pmatrix}
X_{1,\alpha\beta} & 0 \\
0 & Y_{1,\alpha\beta}
\end{pmatrix}
\leq 6\alpha
\begin{pmatrix}
I_1 & -I_1 \\
-I_1 & I_1
\end{pmatrix},
\]
\[
\begin{pmatrix}
0 & 0 \\
0 & 0
\end{pmatrix}
\leq
\begin{pmatrix}
X_{2,\alpha\beta} & 0 \\
0 & Y_{2,\alpha\beta}
\end{pmatrix}
\leq
\begin{pmatrix}
0 & 0 \\
0 & 0
\end{pmatrix},
\]
\[
-6\beta I_4 \leq
\begin{pmatrix}
X_{3,\alpha\beta} & 0 \\
0 & Y_{3,\alpha\beta}
\end{pmatrix}
\leq 6\beta
\begin{pmatrix}
I_2 & -I_2 \\
-I_2 & I_2
\end{pmatrix}.
\]
Consequently, for each \(j \in \{1, 3\}\),
\[
X_{j,\alpha\beta} + Y_{j,\alpha\beta} \leq 0,
\]
and
\begin{equation}
\label{eq:X_Y_equal}
X_{2,\alpha\beta} = Y_{2,\alpha\beta} = 0.
\end{equation}
From these inequalities, by extracting a subsequence if needed, we assume that for each \(k \in \mathbb{N}\), as \(m \to \infty\), the sequence \(\{(X_{j,\alpha_k \beta_m}, Y_{j,\alpha_k \beta_m})\}\) converges to some \((X_{j,k}, Y_{j,k})\). The preceding inequalities then imply
\begin{equation}
X_{1,k} + Y_{1,k} \leq 0, \quad X_{3,k} + Y_{3,k} \leq 0, \quad X_{2,k} = Y_{2,k} = 0.
\label{ref:516}
\end{equation}
Applying \eqref{eq:55}, there exists \((\bar{t}_k, \bar{x}_k, \bar{y}_k) \in Q_T\) and \(\bar{\xi}_k \in \mathbb{R}\) such that
\begin{equation}
\begin{cases}
\displaystyle
\lim_{m \to \infty} (t_{\alpha_k \beta_m}, x_{\alpha_k \beta_m}, y_{\alpha_k \beta_m}, \tau_{\alpha_k \beta_m}, \xi_{\alpha_k \beta_m}, \eta_{\alpha_k \beta_m}) = (\bar{t}_k, \bar{x}_k, \bar{y}_k, \bar{t}_k, \bar{\xi}_k, \bar{y}_k), \\[10pt]
\displaystyle
\lim_{m \to \infty} \beta_m (y_{\alpha_k \beta_m} - \eta_{\alpha_k \beta_m})^2 = 0.
\end{cases}
\label{ref:517}
\end{equation}
Note that \(\bar{x}_k \neq \bar{\xi}_k\). Now we observe that
\begin{align*}
p_{\alpha\beta} &= D_\theta \varphi(\hat{\theta}, \hat{\zeta}) = 2 \bigl( \alpha(x_{\alpha\beta} - \xi_{\alpha\beta}), \beta(y_{\alpha\beta} - \eta_{\alpha\beta}), \beta(t_{\alpha\beta} - \tau_{\alpha\beta}) \bigr), \\
q_{\alpha\beta} &= D_\zeta \varphi(\hat{\theta}, \hat{\zeta}) = -2 \bigl( \alpha(x_{\alpha\beta} - \xi_{\alpha\beta}), \beta(y_{\alpha\beta} - \eta_{\alpha\beta}), \beta(t_{\alpha\beta} - \tau_{\alpha\beta}) \bigr).
\end{align*}
Furthermore,
\begin{align*}
F(y_{\alpha\beta}, p_{\alpha\beta}, \mathcal{X}_{\alpha\beta})
&= -2\beta(t_{\alpha\beta} - \tau_{\alpha\beta}) - \frac{y_{\alpha\beta}}{2} X_{1,\alpha\beta} - \langle \rho \sigma y_{\alpha\beta} X_{2,\alpha\beta}, \textbf{1} \rangle\\
&\quad - \operatorname{Tr}\left(\frac{1}{2} \sigma^2 y_{\alpha\beta} X_{3,\alpha\beta}\right) - 2\alpha \left( r-\delta - \frac{y_{\alpha\beta}}{2} \right) (x_{\alpha\beta} - \xi_{\alpha\beta}) - 2\beta \kappa (\theta - y_{\alpha\beta}) (y_{\alpha\beta} - \eta_{\alpha\beta}), \\[12pt]
F(\eta_{\alpha\beta}, -q_{\alpha\beta}, -\mathcal{Y}_{\alpha\beta})
&= -2\beta(t_{\alpha\beta} - \tau_{\alpha\beta}) + \frac{\eta_{\alpha\beta}}{2} Y_{1,\alpha\beta} + \langle \rho \sigma \eta_{\alpha\beta}  Y_{2,\alpha\beta}, \textbf{1} \rangle \\
&\quad + \operatorname{Tr}\left(\frac{1}{2} \sigma^2 \eta_{\alpha\beta} Y_{3,\alpha\beta}\right) - 2\alpha \left( r-\delta  - \frac{\eta_{\alpha\beta}}{2} \right) (x_{\alpha\beta} - \xi_{\alpha\beta}) - 2\beta \kappa (\theta - \eta_{\alpha\beta}) (y_{\alpha\beta} - \eta_{\alpha\beta}).
\end{align*}

Combining these expressions with \eqref{eq:X_Y_equal}, we have
\begin{equation}
\begin{aligned}
&F(\eta_{\alpha\beta}, -q_{\alpha\beta}, -\mathcal{Y}_{\alpha\beta}) - F(y_{\alpha\beta}, p_{\alpha\beta}, \mathcal{X}_{\alpha\beta}) + r \bigl( v(\tau_{\alpha\beta}, \xi_{\alpha\beta}, \eta_{\alpha\beta}) - u_{\varepsilon_0}(t_{\alpha\beta}, x_{\alpha\beta}, y_{\alpha\beta}) \bigr) \\
&= \frac{\eta_{\alpha\beta} Y_{1,\alpha\beta} + y_{\alpha\beta} X_{1,\alpha\beta}}{2} + \operatorname{Tr}\left(\frac{1}{2} \sigma^2 ( \eta_{\alpha\beta} Y_{3,\alpha\beta} + y_{\alpha\beta} X_{3,\alpha\beta}) \right) \\
&\quad - 2\alpha \left( \frac{y_{\alpha\beta} - \eta_{\alpha\beta}}{2} \right) (x_{\alpha\beta} - \xi_{\alpha\beta}) - 2\beta \kappa (y_{\alpha\beta} - \eta_{\alpha\beta})^2 + r \bigl( v(\tau_{\alpha\beta}, \xi_{\alpha\beta}, \eta_{\alpha\beta}) - u_{\varepsilon_0}(t_{\alpha\beta}, x_{\alpha\beta}, y_{\alpha\beta}) \bigr).
\end{aligned}
\label{eq:f-difference}
\end{equation}

Setting \(\alpha = \alpha_k\) and \(\beta = \beta_m\), and taking the limit as \(m \to \infty\),
we have
\begin{align*}
\gamma \leq{}& \frac{1}{2} \bar{y}_k (X_{1,k} + Y_{1,k}) + \operatorname{Tr}\left( \frac{1}{2} \sigma^2 \bar{y}_k (X_{3,k} + Y_{3,k}) \right) + r \bigl( v(\tau_{k}, \xi_{k}, \eta_{k}) - u_{\varepsilon_0}(t_{k}, x_{k}, y_{k}) \bigr) \\
\leq{}& r \bigl( v(\tau_{k}, \xi_{k}, \eta_{k}) - u_{\varepsilon_0}(t_{k}, x_{k}, y_{k}) \bigr).
\end{align*}
Now, from \eqref{eq:54}, taking \(k \to \infty\) yields a contradiction, and the proof is complete.
\end{proof}

\begin{proposition}
\label{lem:unique_viscosity_solution2}
There exists a unique viscosity solution of Equation \eqref{eq:heston_pde}.
\end{proposition}

\begin{proof}
Assume that there exist two functions \(u\) and \(v\), both of which are viscosity solutions satisfying the conditions in Proposition \ref{lem: u_v_sub_super}, and \(u(T,x,y) = v(T,x,y)\). Since a viscosity solution is both a viscosity subsolution and a viscosity supersolution, it follows that \(u \leq v\) in \(Q_T\), and similarly that \(v \leq u\) in \(Q_T\). Therefore, \(u = v\).
\end{proof}

The uniqueness of the viscosity solution established above is the key link between the probabilistic value function \(P^A\) and the PDE regularity theory that follows. In the next section, we construct a second viscosity solution \(u^*\) via the penalty method, but one that carries additional Sobolev regularity. Since the viscosity solution is unique, \(u^*\) must coincide with \(P^A\), and therefore \(P^A\) inherits the regularity of \(u^*\).

\section{Main regularity results}\label{sec:results}
In this section, we present the main regularity result to prove that the American put option value function under the Heston stochastic volatility model satisfies the smooth-fit property and is \(C^{1,2}\)-regular in the continuation region. Our approach is based on transforming the original obstacle problem into a penalised partial differential equation (PDE), which we solve using a sequence of smooth approximations. We establish that the solutions to the penalised PDE belong to a high-regularity Sobolev space, allowing us to extract a convergent subsequence that yields a viscosity solution to the original problem. By leveraging Sobolev embedding theorems, we demonstrate that this viscosity solution possesses H\"older continuous derivatives, thereby confirming the smooth-fit property.

To facilitate the regularity analysis of \( P^A(t, x, y) \), we reverse time by setting \( \tilde{t} = T - t \), and define \( v(\tilde{t}, x, y) = u(T - \tilde{t}, x, y) \). For notational simplicity, we write \(t\) in place of \(\tilde{t}\) hereafter. Substituting into Equation \eqref{eq:heston_pde}, the problem becomes:
\begin{equation}
\label{eq:heston_variational_inequality}
\begin{cases}
\min \{ (\partial_t - \mathcal{L} + r)v, v - G \} = 0, & (t, x, y) \in (0, T] \times \mathbb{R} \times \mathbb{R}^+, \\
v(0, x, y) = G(x), & (x, y) \in \mathbb{R} \times \mathbb{R}^+,
\end{cases}
\end{equation}
where \( \mathcal{L} \) is the Heston operator defined in~\eqref{eq:heston_operator}. This transformation shifts the terminal condition to an initial condition at \( t = 0 \), aligning with standard parabolic PDE frameworks.

\subsection{Mollification and Penalty Approximation}\label{subsec:penalty}

\begin{lemma}
\label{lem:mollified_obstacle}
There exists a sequence of mollified functions \( G^\epsilon(x) \) converging to \( G(x) \) as \( \epsilon \to 0^+ \), satisfying:
\begin{itemize}
\item[(i)] \(\sup_{x \in \mathbb{R}} |\partial_x G^\epsilon| \leq K\), \(\forall \epsilon > 0\),
\item[(ii)] \( \partial_x G^\epsilon \leq     \partial^{2}_{xx} G^\epsilon \), \(\forall \epsilon > 0\).
\end{itemize}
\end{lemma}

\begin{proof}
The aim is to construct a smooth approximation \(G^\epsilon(x)\) of the obstacle \(G(x)\) by mollification, while ensuring uniform bounds on its derivatives and preserving the structure of the spatial derivatives in regions where \(G\) is smooth.

We begin by defining the standard mollifier
\begin{equation}
\phi(z) =
\begin{cases}
C \exp\left( -\dfrac{1}{1 -  z^2} \right) & \text{if } z^2 < 1, \\
0 & \text{otherwise},
\end{cases}
\label{eq:mollifier_def}
\end{equation}
where \(C > 0\) is the normalization constant chosen such that \(\int_{\mathbb{R}} \phi(z)\, dz = 1\). The rescaled mollifier is
\begin{equation}
\phi_\epsilon(z) = \frac{1}{\epsilon} \phi\left(\frac{z}{\epsilon} \right),
\label{eq:rescaled_mollifier}
\end{equation}
supported in the interval \(|z| < \epsilon\), integrating to 1, and infinitely differentiable. The mollified obstacle is defined as
\begin{equation}
G^\epsilon(x) = (G \ast \phi_\epsilon)(x) = \int_{\mathbb{R}} G(x-z) \phi_\epsilon(z) \, dz.
\label{eq:mollified_function}
\end{equation}
Due to the compact support of \(\phi_\epsilon\), this simplifies to
\begin{equation}
G^\epsilon(x) = \int_{|z| < \epsilon} G(x-z) \phi_\epsilon(z)\, dz.
\label{eq:mollified_function_simplified}
\end{equation}
Using the explicit expression
\[
G(x-z) =
\begin{cases}
K - e^{x-z} & \text{if } z > x - \ln K, \\
0 & \text{if } z \leq x - \ln K,
\end{cases}
\]
we obtain
\begin{equation}
G^\epsilon(x) = \int_{\max(x - \ln K, -\epsilon)}^{\epsilon } \bigl(K - e^{x-z}\bigr) \phi_\epsilon(z) \, dz .
\label{eq:mollified_explicit}
\end{equation}

Differentiating under the integral sign and applying the Leibniz rule, distinguishing the cases
\[ x > \ln K - \epsilon \quad \text{and} \quad x \leq \ln K - \epsilon,\]
we find
\begin{equation}
\partial_x G^\epsilon(x) =  \int_{\max(x - \ln K , -\epsilon)}^{\epsilon} -e^{x-z } \phi_\epsilon(z) \, dz.
\label{eq:first_derive}
\end{equation}
For the second derivative,
\begin{equation}
\partial_{xx}^2 G^\epsilon(x) =
\begin{cases}
\displaystyle \int_{-\epsilon}^{\epsilon}
-e^{x - z } \phi_\epsilon(z) \, dz,
& \text{if } x \le \ln K - \epsilon, \\[1.5em]
\begin{aligned}
& \displaystyle \int_{x - \ln K }^{\epsilon}
-e^{x - z} \phi_\epsilon(z) \, dz  \\
& +  K \phi_\epsilon( x - \ln K) \,
\end{aligned}
& \text{if } x > \ln K - \epsilon.
\end{cases}
\label{eq:second_derive}
\end{equation}
Thus,
\begin{equation}
\begin{cases}
\partial_{xx}^2 G^\epsilon(x) = \partial_x G^\epsilon(x)
& \text{if } x \leq \ln K - \epsilon, \\[1em]
\partial_{xx}^2 G^\epsilon(x) = \partial_x G^\epsilon(x) +  K \phi_\epsilon( x - \ln K ) \,
& \text{if } x > \ln K - \epsilon.
\end{cases}
\label{eq:compare_g_derive}
\end{equation}
From \eqref{eq:compare_g_derive}, since \(\phi_\epsilon \geq 0\), it follows that \(\partial_x G^\epsilon \leq \partial_{xx}^2 G^\epsilon\). For the bound on the first derivative, note that the integration domain in \eqref{eq:first_derive} satisfies \(z \geq x - \ln K\), so that \(x - z \leq \ln K\) and hence \(e^{x-z} \leq K\) throughout the integral. Since \(\phi_\epsilon \geq 0\) and \(\int \phi_\epsilon = 1\), it follows that \(\sup_{x \in \mathbb{R}} |\partial_x G^\epsilon| \leq K\).
\end{proof}

\begin{lemma}
\label{lem:g_epsilon_con}
Let $g(x)$ be a function on $\mathbb{R}$ satisfying the following conditions:
\begin{itemize}
    \item[(i)] \(g(x) \geq 0\), \(\forall x \in \mathbb{R} \),
    \item[(ii)] \( |\partial_x g| \) is bounded uniformly by some constant \(M\) (in particular, \(g\) is Lipschitz continuous),
    \item[(iii)] \(\partial_x g \leq \partial^2_{xx} g\),
    \item[(iv)] \(G \leq g \leq K + 1\).
\end{itemize}
Then for each \(\epsilon > 0\), let \(g^{\epsilon}\) be the mollified sequence of \(g\) constructed in the same manner as in Lemma~\ref{lem:mollified_obstacle}. Then \(g^{\epsilon}\) satisfies:
\begin{itemize}
    \item[(i)] \(g^{\epsilon} \geq 0\) on \(\mathbb{R}\),
    \item[(ii)] \( |\partial_x g^{\epsilon}| \) is bounded by \(M\),
    \item[(iii)] \(\partial_x g^{\epsilon} \leq \partial^2_{xx} g^{\epsilon}\),
    \item[(iv)] \(G^{\epsilon} \leq g^{\epsilon} \leq K + 1\).
\end{itemize}
\end{lemma}
\begin{proof}
    Referring to the arguments in Lemma~\ref{lem:mollified_obstacle}, we can easily verify.
\end{proof}

\begin{remark}
    We introduce a function $g$
serving as a \emph{template} for the boundary data: $g$ is any function
satisfying the conditions listed in Lemma \ref{lem:g_epsilon_con}, and $g^\epsilon$ denotes
its mollification in the sense of Lemma \ref{lem:mollified_obstacle}. A concrete choice of $g$ will be
specified later in the proof of Proposition~\ref{lem:subsequence_existence}, where we set
$g = g_n$ for each bounded subdomain $D^n_T$.
\end{remark}

\subsection{Auxiliary PDE Tools: Maximum Principles and Barriers}\label{subsec:auxiliary}
In this subsection, we provide foundational PDE tools, such as maximum principles and barriers, to establish solvability and bounds for the penalised equation.

\begin{lemma}\label{lem:nonnegative_v}
Let $D \subset \mathbb{R} \times \mathbb{R}^+$ be a bounded open subset, and let $f(t, x, y) \geq 0$ on $\overline{D_T}$. If $v \in C(\overline{D_T}) \cap C^{1,2}(D_T)$ satisfies
\begin{enumerate}[label=(\roman*)]
  \item $v(t, x, y) \geq 0$ for all $(t, x, y) \in \partial_{P} D_T$,
  \item $(\partial_t - \mathcal{L} + f)\, v(t, x, y) \geq 0$ for all $(t, x, y) \in D_T$,
\end{enumerate}
then $v(t, x, y) \geq 0$ on $\overline{D_T}$.
\end{lemma}

\begin{proof}
Suppose, for the sake of contradiction, that $v$ takes a negative value at some point in $D_T$. Since $D$ is bounded, $\overline{D_T}$ is compact. As $v \in C(\overline{D_T})$, the function $v$ attains its minimum on $\overline{D_T}$. Because $v \geq 0$ on the parabolic boundary $\partial_P D_T$ and \(v\) is bounded below, the minimum must be attained at some interior point $(t_0, x_0, y_0) \in D_T$ with
\[
  v(t_0, x_0, y_0) < 0.
\]

Since $(t_0, x_0, y_0)$ is a global minimum of $v$ on $\overline{D_T}$, the standard necessary conditions for an interior minimum of a $C^{1,2}$ function yield:
\begin{align}
  &\partial_t v(t_0, x_0, y_0) \leq 0, \label{eq:min_dt}\\[4pt]
  &\partial_x v(t_0, x_0, y_0) = 0, \qquad \partial_y v(t_0, x_0, y_0) = 0, \label{eq:min_grad}\\[4pt]
  &D^2 v(t_0, x_0, y_0) \text{ is nonnegative definite}. \label{eq:min_hessian}
\end{align}
Note that the condition $\partial_t v \leq 0$ in \eqref{eq:min_dt} accounts for the possibility that $t_0 = T$; if $t_0 \in (0, T)$, then $\partial_t v(t_0, x_0, y_0) = 0$. We now evaluate each component of $(\partial_t - \mathcal{L} + f)\,v$ at $(t_0, x_0, y_0)$.

Since $\partial_x v = \partial_y v = 0$ at $(t_0, x_0, y_0)$, the first-order terms of the Heston operator vanish. The second-order part of $\mathcal{L}v$ can be written as $\mathrm{Tr}\bigl(a(t_0, x_0, y_0)\, D^2 v(t_0, x_0, y_0)\bigr)$, where the diffusion matrix is
\[
  a(t_0, x_0, y_0) = \frac{y_0}{2}
  \begin{pmatrix} 1 & \rho\sigma \\[2pt] \rho\sigma & \sigma^2 \end{pmatrix}.
\]
Since $D \subset \mathbb{R} \times \mathbb{R}^+$, we have $y_0 > 0$. Moreover, the matrix $\begin{pmatrix} 1 & \rho\sigma \\ \rho\sigma & \sigma^2 \end{pmatrix}$ is nonnegative definite because its eigenvalues are nonnegative (its determinant is $\sigma^2(1 - \rho^2) \geq 0$ and its trace is $1 + \sigma^2 > 0$). Therefore, $a(t_0, x_0, y_0)$ is nonnegative definite. By \eqref{eq:min_hessian}, $D^2 v(t_0, x_0, y_0)$ is also nonnegative definite. Since the trace of the product of two nonnegative definite symmetric matrices is nonnegative, we conclude that
\[
  \mathrm{Tr}\bigl(a(t_0, x_0, y_0)\, D^2 v(t_0, x_0, y_0)\bigr) \geq 0,
\]
and hence $\mathcal{L}v(t_0, x_0, y_0) \geq 0$.
Collecting the estimates above, we obtain
\[
  (\partial_t - \mathcal{L} + f)\, v(t_0, x_0, y_0)
  = \underbrace{\partial_t v(t_0, x_0, y_0)}_{\leq\, 0}
  - \underbrace{\mathcal{L}v(t_0, x_0, y_0)}_{\geq\, 0}
  + \underbrace{f(t_0, x_0, y_0) \cdot v(t_0, x_0, y_0)}_{\leq\, 0}
  \;\leq\; 0.
\]
This contradicts assumption (ii). Therefore, no such negative value can exist, and $v \geq 0$ on $\overline{D_T}$.
\end{proof}

\begin{definition}[Barrier function]
\label{def:barrier}
Let \((\bar{t},\bar{z}) \in \partial_P D_T\). A function \(w\) is called a \emph{barrier function for the operator
\(\partial_t - \mathcal{L} + r\) at \((\bar{t},\bar{z})\)} if there exists an open set \(V \subset \mathbb{R}^3\)
with \((\bar{t},\bar{z}) \in V\) such that \(w \in C^2(V \cap \overline{D_T};\, \mathbb{R})\) and
\begin{itemize}
\item[(i)]  \(-\bigl(\partial_t - \mathcal{L} + r\bigr) w \leq -1\) in \(V \cap D_T\);
\item[(ii)] \(w > 0\) in \(\bigl(V \cap \overline{D_T}\bigr) \setminus \{(\bar{t},\bar{z})\}\) \quad and \quad \(w(\bar{t},\bar{z}) = 0\).
\end{itemize}
\end{definition}

\begin{proposition}
\label{lem: barrier_existence}
Provided that the domain \( D \) is bounded away from the set \( \{ y = 0 \} \) and  \(\partial_P D\) belongs to the class \(H^{\beta+2}\) with \(\beta \in (0, 1)\), there exists a barrier function for the operator \( \partial_t - \mathcal{L} + r \) at every point \( (t, z) \in \partial_P D_T \).
\end{proposition}

\begin{proof}
Since $D$ is bounded away from $\{y = 0\}$, there exists $\epsilon > 0$ such that $y \geq \epsilon$ for all $(x,y) \in D$. Therefore, the Heston operator
\begin{equation}
\mathcal{L} = \frac{y}{2} \left( \frac{\partial^2}{\partial x^2} + 2 \rho \sigma \frac{\partial^2}{\partial x \partial y} + \sigma^2 \frac{\partial^2}{\partial y^2} \right) + \left( r-\delta - \frac{y}{2} \right) \frac{\partial}{\partial x} + \kappa (\theta - y) \frac{\partial}{\partial y}
\end{equation}
is uniformly parabolic on $D_T$. We consider two cases.

\noindent\textbf{Case 1:}
Suppose the boundary point is of the form \((0, \bar{x}, \bar{y})\). Fix any \(R > 0\) and let
\begin{equation}
\label{eq: V_neigbor_defined}
  V \;:=\; B_R(0, \bar{x}, \bar{y})
\;=\; \bigl\{ (t,x,y) \in \mathbb{R}^3 \;:\; t^2 + (x - \bar{x})^2 + (y - \bar{y})^2 < R^2 \bigr\},
\end{equation}
which is an open neighbourhood of \((0, \bar{x}, \bar{y})\) in \(\mathbb{R}^3\) in the sense of Definition~\ref{def:barrier}.
Define
\begin{equation}
\label{eq:barrier_case1}
w(t, x, y) \;=\; e^{r t} \left[ (x - \bar{x})^2 + (y - \bar{y})^2 + C t \right] \qquad \text{on } V \cap \overline{D_T},
\end{equation}
where \(C > 0\) is a constant to be determined below. Clearly \(w \in C^2(V \cap \overline{D_T})\), \(w(0, \bar{x}, \bar{y}) = 0\), and \(w > 0\) on \((V \cap \overline{D_T}) \setminus \{(0, \bar{x}, \bar{y})\}\), so condition~(ii) of Definition~\ref{def:barrier} holds. A direct computation yields
\begin{align*}
- \left( \partial_t - \mathcal{L} + r \right) w
&= e^{r t} \bigg[\, y(1 + \sigma^2) + 2\left( r - \delta - \tfrac{y}{2} \right)(x - \bar{x}) \\
& \qquad\qquad + 2\kappa (\theta - y)(y - \bar{y}) - C \,\bigg] - 2 r w.
\end{align*}
From \eqref{eq: V_neigbor_defined}, the bracketed expression (without the \(-C\) term) is bounded above by
\[
M(R) \;:=\; (\bar{y} + R)(1 + \sigma^2) + \bigl(2|r - \delta| + \bar{y} + R\bigr)\, R + 2\kappa\,(\theta + \bar{y} + R)\, R,
\]
which depends only on \(R\) and the model parameters. Since \(-2rw \leq 0\), it follows that
\[
- \left( \partial_t - \mathcal{L} + r \right) w \;\leq\; e^{rt}\bigl[\, M(R) - C \,\bigr].
\]
Now from the fact that \(e^{rt} \geq 1\), we can select \(C\) sufficiently large to ensure that \(- \left( \partial_t - \mathcal{L} + r \right) w \leq -1\) on \(V \cap D_T\), so condition~(i) of Definition~\ref{def:barrier} holds. Therefore \(w\) is a barrier function for \(\partial_t - \mathcal{L} + r\) at \((0, \bar{x}, \bar{y})\), with associated open neighbourhood \(V\).

\noindent\textbf{Case 2:}
If the boundary point is of the form \(( \bar{t}, \bar{z} ) \in \partial_P D_T\) with \(\bar{t} \in (0,T]\), let \((\tilde{t}, \tilde{z})\) be the centre of a sphere that is externally tangent to the cylinder at \((\bar{t}, \bar{z})\), and set \(R_0 := |\bar{z} - \tilde{z}|\). Let
\begin{equation}
\label{eq: V_neigbor_case2}
V \;:=\; B_1(\bar{t}, \bar{z})
\;=\; \bigl\{ (t,z) \in \mathbb{R}^3 \;:\; (t - \bar{t})^2 + |z - \bar{z}|^2 < 1 \bigr\}.
\end{equation}
Define, on $V \cap \overline{D_T}$,
\[
w(t,z) \;=\; C\, e^{r t} \left( \frac{1}{R_0^{\,p}} - \frac{1}{R^{p}} \right),
\qquad R \;:=\; \bigl( |z - \tilde{z}|^2 + (t - \bar{t})^2 \bigr)^{1/2},
\]
where $C, p > 0$ will be determined below. Since $R \geq R_0$ on
$V \cap \overline{D_T}$ with equality only at $(\bar{t}, \bar{z})$, we have
$w(\bar{t}, \bar{z}) = 0$ and $w > 0$ elsewhere on $V \cap \overline{D_T}$, so
condition~(ii) of Definition~\ref{def:barrier} holds. A direct computation yields
\begin{align*}
-\bigl(\partial_t - \mathcal{L} + r\bigr) w
&= \frac{C p\, e^{r t}}{R^{p+4}} \bigg(\! -\frac{p+2}{2} \sum_{i,j=1}^{2}\sigma_{ij}(z_i - \tilde{z}_i)(z_j - \tilde{z}_j) \\
&\qquad + \frac{R^2}{2} \sum_{i=1}^{2} \sigma_{ii}
+ R^2 \sum_{i=1}^{2} b_i (z_i - \tilde{z}_i)
- (t - \bar{t})\, R^2 \bigg) \;-\; 2 r w,
\end{align*}
where \(z_i\) (\(i = 1, 2 \)) denotes the \(i\)-th coordinate of \(z \in D\), \(\sigma_{ij}\) denotes the \((i,j)\)-entry of the matrix \(\sigma\sigma^T\), and \(\sigma, b\) are defined in \eqref{eq:def_of_signma_b}. By uniform ellipticity
of $\mathcal{L}$ on $\overline{D_T}$, there exists $\lambda > 0$ with
\[
  \sum_{i,j=1}^{2} \sigma_{ij}(z_i - \tilde{z}_i)(z_j - \tilde{z}_j)
  \;\geq\; \lambda\,|z - \tilde{z}|^2 \;\geq\; \lambda\,R_0^{\,2}
  \qquad\text{on } V \cap D_T.
\]
Moreover, on $V$ we have $|z - \tilde{z}| \leq R_0 + 1$ and $|t - \bar{t}| \leq 1$,
so the remaining three terms in the bracket are bounded above by a constant
$K = K(R_0,\sigma,b) > 0$ depending only on $R_0$ and the model parameters.
Hence, choosing
\begin{equation}
\label{eq:choice_of_p}
  p \;>\; \frac{2K}{\lambda\, R_0^{\,2}} \;-\; 2,
\end{equation}
the bracket is bounded above by a strictly negative constant on $V \cap D_T$.
Since $-2rw \leq 0$, taking $C > 0$ sufficiently large yields
$-(\partial_t - \mathcal{L} + r) w \leq -1$ on $V \cap D_T$, so condition~(i)
of Definition~\ref{def:barrier} holds. Therefore $w$ is a barrier function for
$\partial_t - \mathcal{L} + r$ at $(\bar{t}, \bar{z})$, with associated open
neighbourhood $V$.
\end{proof}

\subsection{Existence and uniform bound of the solution}\label{subsec:existence}

To apply the penalty method to the problem \eqref{eq:heston_variational_inequality}, we introduce a penalised problem for each \( \epsilon \in (0, 1) \). For each bounded open subset \(D \subset \mathbb{R} \times \mathbb{R}^+\) whose parabolic boundary \(\partial_P D\) belongs to the class \(H^{\beta+2}\) with \(\beta \in (0, 1)\) and does not intersect the line \(y=0\), consider the following equation:
\begin{equation}
\label{eq:penalty_equation}
\begin{cases}
(\partial_t - \mathcal{L} + r) v + p_\epsilon(v - G^\epsilon) = 0, & (t, x, y) \in D_T, \\
v|_{\partial_P D_T} = g^{\epsilon}(x),
\end{cases}
\end{equation}
where \( \{G^\epsilon\}_{\epsilon \in (0,1)} \) is the mollified sequence from Lemma~\ref{lem:mollified_obstacle}, \( g^{\epsilon} \) is defined in Lemma~\ref{lem:g_epsilon_con}, and the penalty function \( p_\epsilon(\xi) \in C^\infty(\mathbb{R}) \) satisfies:
\begin{align}
    &\text{(i)} \ p_\epsilon(\xi) \leq 0, \quad
    \text{(ii)} \ p_\epsilon(\xi) = 0 \text{ for } \xi \geq \epsilon, \quad
    \text{(iii)} \ p_\epsilon(0) = - (\left| r-\delta \right| K + r K), \notag \\
    &\text{(iv)} \ p'_\epsilon(\xi) \geq 0, \quad
    \text{(v)} \ p''_\epsilon(\xi) \leq 0, \quad \text{and} \quad
    \text{(vi)} \ \lim_{\epsilon \downarrow 0} p_\epsilon(\xi) =
    \begin{cases}
        0, & \xi > 0, \\
        -\infty, & \xi < 0.
    \end{cases} \label{eq:43_hxing}
\end{align}
Note that \(p_{\epsilon}\) can be chosen as a smooth mollification of the function \(\min(-2p_{\epsilon}(0)\xi/\epsilon + p_{\epsilon}(0), 0)\).

\begin{proposition}
\label{lem:penalty_var}
Assume the following conditions hold:
\begin{itemize}
    \item[(i)] The function \(h = h(t,x,y, v) \in \operatorname{Lip}\left( \overline{D_T} \times \mathbb{R} \right)\),
    \item[(ii)] The boundary data $\overline{g} \in C^{\infty}(\mathbb{R}) $ and bounded in \(\mathbb{R}\).
\end{itemize}
Then there exists a classical solution $v \in H^{2+\beta,1+\beta/2}(\overline{D_T})$, for every $\beta \in (0,1)$, to the parabolic problem:
\begin{equation}\label{eq:main-pde}
\begin{cases}
-(\partial_t - \mathcal{L} + r)v = h(\cdot, v), & \text{in } D_T, \\[0.5em]
v|_{\partial_P D_T} = \overline{g}(x).
\end{cases}
\end{equation}
\end{proposition}

\begin{proof}
We construct a solution via a recursive sequence. Choose a constant \(c > 0\) such that
\begin{equation}
|h(t,x,y,v)| \leq c(1 + |v|) \quad \text{for all } (t,x,y,v) \in \overline{D_T} \times \mathbb{R}.
\end{equation}
We set \(v_0(t,x,y) = e^{ct} (1 + \|\overline{g}\|_{L_\infty}) - 1\) and consider the following problem for each \(j \geq 1\),
\begin{equation}
\label{eq:recursive-seq}
\begin{cases}
-(\partial_t - \mathcal{L} + r)v_j - \lambda v_j = h(\cdot, v_{j-1}) - \lambda v_{j-1} & \text{in } D_T, \\
v_j|_{\partial_P D_T} = \overline{g},
\end{cases}
\end{equation}
where \(\lambda > 0\) denotes the Lipschitz constant of \(h\) with respect to the last variable. By Theorem~5.2 in~\cite{Ladyzenskaja1968} (p.~320), there exists a unique solution \(v_j \in H^{2+\beta,1+\beta/2}(\overline{D_T})\) for every \(\beta \in (0, 1)\). We then prove by induction that the sequence \((v_j)_{j \geq 0}\) is decreasing and uniformly bounded from below.

For the base case, we consider \(w = v_1 - v_0\); then
\begin{align}
-(\partial_t - \mathcal{L} + r)w - \lambda w &= h(\cdot, v_0) - \lambda v_0 + (\partial_t - \mathcal{L} + r)v_0 + \lambda v_0 \notag \\
&= h(\cdot, v_0) + (\partial_t - \mathcal{L} + r)v_0 \notag \\
&= h(\cdot, v_0) + c(v_0 + 1) + r v_0 \geq 0 \quad \text{in } D_T.
\end{align}
By the choice of $v_0$, we have $v_1 = \overline{g} \leq v_0$ on $\partial_P D_T$. Consequently, Lemma~\ref{lem:nonnegative_v} yields $v_1 \leq v_0$ throughout $\overline{D_T}$. Next, we assume \(v_j \leq v_{j-1}\) in \(\overline{D_T}\). Let \(w = v_{j+1} - v_j\). Then
\begin{align}
-(\partial_t - \mathcal{L} + r)w - \lambda w &= h(\cdot, v_j) - h(\cdot, v_{j-1}) - \lambda (v_j - v_{j-1}) \notag \\
&\geq 0 \quad \text{in } D_T,
\end{align}
where the inequality follows from the Lipschitz continuity of \(h\) and the induction hypothesis \(v_j \leq v_{j-1}\). Moreover, \(w = 0\) on \(\partial_P D_T\). Applying Lemma~\ref{lem:nonnegative_v} yields \(v_{j+1} \leq v_j\) in \(\overline{D_T}\). A similar induction (considering \(-v_j\)) shows \(v_j \geq -v_0\) in \(\overline{D_T}\). Thus,
\begin{equation}
-v_0 \leq v_{j+1} \leq v_j \leq v_0 \quad \text{in } \overline{D_T}
\end{equation}
for all \(j \geq 0\). By the monotone convergence theorem, there exists a function \(v\) such that \(v_j \to v\) pointwise in \(\overline{D_T}\) as \(j \to \infty\), and \(-v_0 \leq v \leq v_0\). By Theorem~5.2 in~\cite{Ladyzenskaja1968}, we have
\[
\|v_j\|_{\overline{D_T}}^{(\beta+2)} \leq c \bigl( \|h(\cdot, v_{j-1})\|_{D_T}^{(\beta)} + \|\bar{g}\|_{\partial_P D_T}^{(\beta+2)} \bigr) \leq C,
\]
where \(C\) is independent of \(j\), since \(v_{j-1}\) is uniformly bounded and \(h\) is Lipschitz. By the Arzelà--Ascoli theorem (embedded compactly), there exists a subsequence (relabelled \(v_j\)) converges locally in \(H^{2+\beta,1+\beta/2}(D_T)\). Taking the limit in \eqref{eq:recursive-seq} along this subsequence, we obtain
\begin{equation}
\begin{cases}
-(\partial_t - \mathcal{L} + r)v = h(\cdot, v) & \text{in } D_T, \\
v|_{\partial_P D_T} = \overline{g}.
\end{cases}
\end{equation}
\textbf{Continuity up to boundary}

We now prove that \(v \in C(\overline{D_T})\); the idea of the proof is to utilise the barrier function. Given \(\bar{\xi} = (\bar{t}, \bar{x},\bar{y}) \in \partial_P D_T\) and \(\varepsilon > 0\), we consider an open neighbourhood \(V\) of \(\bar{\xi}\) such that
\[
|\bar{g}(x) - \bar{g}(\bar{x})| \leq \varepsilon, \quad \forall \xi = (t,x,y) \in V \cap \partial_P D_T.
\]
From Proposition \ref{lem: barrier_existence}, there exists a barrier function \(w\) in \(V \cap D_T\). We define
\[
v^\pm(\xi) = \bar{g}(\bar{x}) \pm (\varepsilon + k_\varepsilon w(\xi)),
\]
where \(k_\varepsilon > 0\) is a sufficiently large constant, independent of \(j\), chosen such that
\[
-(\partial_t - \mathcal{L} + r)(v_j - v^+) \geq h(\cdot, v_{j-1}) - \lambda (v_{j-1} - v_j) + k_\varepsilon \geq 0 \quad \text{in } V \cap D_T,
\]
and \(v_j \leq v^+\) on \(\partial(V \cap D_T)\). By Lemma \ref{lem:nonnegative_v}, we obtain \(v_j \leq v^+\) in \(V \cap D_T\). Analogously, \(v_j \geq v^-\) in \(V \cap D_T\). Passing to the limit as \(j \to \infty\), we deduce
\[
\bar{g}(\bar{x}) - \varepsilon - k_\varepsilon w(\xi) \leq v(\xi) \leq \bar{g}(\bar{x}) + \varepsilon + k_\varepsilon w(\xi), \quad \forall \xi \in V \cap D_T.
\]
Thus,
\[
\bar{g}(\bar{x}) - \varepsilon \leq \liminf_{\xi \to \bar{\xi}} v(\xi) \leq \limsup_{\xi \to \bar{\xi}} v(\xi) \leq \bar{g}(\bar{x}) + \varepsilon, \quad \forall\xi \in V \cap D_T.
\]
Since \(\varepsilon > 0\) is arbitrary, this implies \(v(\xi) \to g(\bar{x})\) as \(\xi \to \bar{\xi}\) along the parabolic boundary, establishing the continuity of \(v\) up to \(\partial_P D_T\).
\end{proof}

\begin{remark}
In this paper, rather than solving Equation~\eqref{eq:penalty_equation} over the entire space \(\mathbb{R} \times \mathbb{R}^+\) as in~\cite{haoxing2012regularity}, we first address it within bounded open domains \(D \Subset \mathbb{R} \times \mathbb{R}^+\) excluding the line \(y=0\). This approach is adopted because the operator \(\mathcal{L}\) is uniformly elliptic in such domains but degenerates along \(y=0\).
\end{remark}

\begin{proposition}
\label{lem:bounds_v_epsilon}
\label{lem:bound_epsilon}
For each \(\epsilon \) in \( (0,1) \),
the Equation \eqref{eq:penalty_equation} has a solution $v^{\epsilon} \in H^{2 + \beta, 1 + \frac{\beta}{2}}(\overline{D_T})$ that satisfies:
\begin{enumerate}[label=\alph*)]
    \item \textit{\( v^{\epsilon} \) is the unique bounded classical solution of Equation \eqref{eq:penalty_equation},}
    \item \textit{\( 0 \leq v^{\epsilon}(t,x,y) \leq K+1 \) on \(\overline{D_T}\),}
    \item \textit{\( v^{\epsilon}(t,x,y) \geq G^{\epsilon}(x) \) on \( \overline{D_T} \),}
    \item \textit{\( p_\epsilon \left(v^{\epsilon}(t,x,y) - G^{\epsilon}(x)\right) \) is bounded uniformly in \( \epsilon \) and in \(D\).}
\end{enumerate}
\end{proposition}

\begin{proof}
We establish each property of \( v^\epsilon \) systematically, using the structure of the penalised equation, the maximum principle, and derivative estimates. The proof proceeds by addressing each property in turn, labelled as parts a) through d), to establish the desired results for the solution \( v^{\epsilon} \).
\begin{enumerate}[label=\alph*)]
\item \textbf{Uniqueness of \(v^{\epsilon}\) as a bounded classical solution}. We replace $h(\cdot, v)$, $\overline{g}$ in Proposition \ref{lem:penalty_var} by \(p_{\epsilon}(v - G^{\epsilon})\), \(g^{\epsilon}\) respectively
to infer a solution $v^{\epsilon} \in H^{2+\beta,1+\frac{\beta}{2}}(\overline{D_T})$ in Equation \eqref{eq:penalty_equation}. To demonstrate uniqueness, suppose \( v_1 \) and \( v_2 \) are two distinct bounded classical solutions to Equation \eqref{eq:penalty_equation}. Then their difference \( v_1 - v_2 \) satisfies the following equation:
\[
\begin{cases}
(\partial_t - \mathcal{L} + r)(v_1 - v_2) + p_\epsilon(v_1 - G^\epsilon) - p_\epsilon(v_2 - G^\epsilon) = 0, & (t,x,y) \in D_T, \\
(v_1 - v_2)|_{\partial_P D_T} = 0.
\end{cases}
\]
By the mean value theorem, we can express the penalty term difference as:
\[
p_\epsilon(v_1 - G^\epsilon) - p_\epsilon(v_2 - G^\epsilon) = p_\epsilon'(\xi)(v_1 - v_2),
\]
where \( \xi \in \mathbb{R} \) and \( p_\epsilon'(\xi) \geq 0 \). Applying Lemma~\ref{lem:nonnegative_v} with \( f = r + p_\epsilon'(\xi) \) (note that Lemma~\ref{lem:nonnegative_v} remains applicable here, even though $f$ depends on $v_1 - v_2$), we conclude that \( v_1 \geq v_2 \) on \( \overline{D_T} \). Repeating the argument for \( v_2 - v_1 \), we obtain the reverse inequality, thus proving \( v_1 = v_2 \).

\item \textbf{Uniform bound on \(v^\epsilon\)}.
To establish uniform bounds on \(v^\epsilon\), observe that Equation~\eqref{eq:penalty_equation} implies
\[
(\partial_t - \mathcal{L} + r)v^\epsilon = -p_\epsilon(v^\epsilon - G^\epsilon) \geq 0 \quad \text{in } D_T,
\]
together with the boundary condition \(v^\epsilon|_{\partial_P D_T} = g^\epsilon \geq 0\).
By Lemma~\ref{lem:nonnegative_v}, it follows that \(v^\epsilon(t,x,y) \geq 0\) on \(\overline{D_T}\). Next, consider the auxiliary function \(w = K + 1 - v^\epsilon\). Then \(w\) satisfies
\[
(\partial_t - \mathcal{L} + r)w = r(K + 1) + p_\epsilon(v^\epsilon - G^\epsilon) \quad \text{in } D_T.
\]
Since \(G^\epsilon(x) \leq K\), we have \(p_\epsilon(K + 1 - G^\epsilon) = 0\). Applying the mean value theorem, there exists some \(\xi \in \mathbb{R}\), such that
\[
(\partial_t - \mathcal{L} + r)w + p_\epsilon(K + 1 - G^\epsilon) - p_\epsilon(v^\epsilon - G^\epsilon) = \bigl[\partial_t - \mathcal{L} + r + p_\epsilon'(\xi)\bigr]w = r(K + 1) \geq 0.
\]
Moreover, by Lemma~\ref{lem:g_epsilon_con}, the boundary condition reads
\[
w|_{\partial_P D_T} = K + 1 - g^\epsilon(x) \geq 0.
\]
An application of Lemma~\ref{lem:nonnegative_v} with the nonnegative term \(f = r + p_\epsilon'(\xi) \geq 0\) yields \(w(t,x,y) \geq 0\) on \(\overline{D_T}\) for every \(\epsilon \in (0,1)\). Therefore,
\[
0 \leq v^\epsilon(t,x,y) \leq K + 1 \quad \text{for all } (t,x,y) \in \overline{D_T} \text{ and } \epsilon \in (0,1).
\]
\item \textbf{Verify \(v^\epsilon(t,x,y) \geq G^\epsilon(x)\)}.
We first derive an upper bound for \((\partial_t - \mathcal{L} + r) G^\epsilon\). Direct computation yields
\begin{align}
(\partial_t - \mathcal{L} + r) G^\epsilon(x)
&=  - \frac{y}{2} \partial_{xx}^2 G^\epsilon
+ \left( \frac{y}{2} - r+\delta \right) \partial_x G^\epsilon + r G^\epsilon \notag \\
&\leq  - (r-\delta) \partial_x G^\epsilon + r G^\epsilon \notag \\
&\leq \left| r-\delta \right| K + r K = -p_\epsilon(0),
\end{align}
with the first and second inequalities following from Lemma~\ref{lem:mollified_obstacle}. Thus,
\[
(\partial_t - \mathcal{L} + r)(v^\epsilon - G^\epsilon)
= -p_\epsilon(v^\epsilon - G^\epsilon) - (\partial_t - \mathcal{L} + r) G^\epsilon
\geq -p_\epsilon(v^\epsilon - G^\epsilon) + p_\epsilon(0).
\]
By the mean value theorem, there exists \(\xi\) such that
\[
(\partial_t - \mathcal{L} + r + p_\epsilon'(\xi))(v^\epsilon - G^\epsilon) \geq 0 \quad \text{in } D_T.
\]
Moreover, Lemma~\ref{lem:g_epsilon_con} guarantees
\[
(v^\epsilon - G^\epsilon)|_{\partial_P D_T} = (g^\epsilon - G^\epsilon) \geq 0.
\]
Applying Lemma~\ref{lem:nonnegative_v} with the nonnegative term \(f = r + p_\epsilon'(\xi) \geq 0\) therefore implies
\[
v^\epsilon(t,x,y) \geq G^\epsilon(x) \quad \text{in } \overline{D_T}.
\]

\item
From part c), we have \( v^{\epsilon} - G^{\epsilon} \geq 0 \). Since \( p_{\epsilon} \) is nondecreasing (\eqref{eq:43_hxing}~(iv)) and \( p_{\epsilon} \leq 0 \) (\eqref{eq:43_hxing}~(i)), it follows that
\[
p_{\epsilon}(0) \leq p_{\epsilon}(v^{\epsilon} - G^{\epsilon}) \leq 0 \quad \text{on } \overline{D_T}.
\]
By \eqref{eq:43_hxing}~(iii), the value \( p_{\epsilon}(0) = -(|r - \delta|K + rK) \) depends only on the model parameters and is independent of both \( \epsilon \) and the domain \( D \). Therefore, \( |p_{\epsilon}(v^{\epsilon} - G^{\epsilon})| \leq |r - \delta|K + rK \) uniformly in \( \epsilon \in (0,1) \) and \( D \).
\end{enumerate}
\end{proof}

\begin{proposition}
\label{lem:cauchy_problem_solution}
Let $v \in W^{2,1}_{p,\mathrm{loc}}(D_T) \cap C(\overline{D_T})$ be a solution to the Cauchy problem
\begin{equation}
\begin{cases}
(\partial_t - \mathcal{L} + r) v = f(t, x, y), & (t, x, y) \in D_T, \\
v|_{\partial_P D_T} = g(x),
\end{cases}
\end{equation}
where $f \in L_{p,\mathrm{loc}}(D_T)$. For any bounded domains $\mathcal{O}^{(1)} \Subset \mathcal{O}^{(2)} \Subset D$, let $d(\mathcal{O}^{(1)}, \partial \mathcal{O}^{(2)})$ denote the infimum of the distances from points in $\mathcal{O}^{(1)}$ to the boundary of $\mathcal{O}^{(2)}$. Suppose $\eta > 0$ satisfies $d(\mathcal{O}^{(1)}, \partial \mathcal{O}^{(2)}) > \eta$. Then, for every $p \ge 3$, there exists a positive constant $C$, depending only on $p$, $\mathcal{O}^{(1)}$, $\mathcal{O}^{(2)}$, and $\eta$ (in particular, $C$ is independent of $D$), such that
\[
\|v\|_{W^{2,1}_p(\mathcal{O}^{(1)}_T)} \leq C \left( \|v\|_{L_p(\mathcal{O}^{(2)}_T)} + \|f\|_{L_p(\mathcal{O}^{(2)}_T)} \right).
\]
\end{proposition}
\begin{proof}
In this proof, we denote by \(d(z_0)\) the distance from a point \(z_0\) to the boundary \(\partial \mathcal{O}^{(2)}\). The parabolic cylinder is defined as \(Q(R, t_0, z_0) = \{ (t, z) \in \mathbb{R}^3 : |z - z_0| + |t-t_0| < R, \, t < t_0 \}\). Finally, let \(\lambda\) denote the bound on the coefficients of \(\mathcal{L}\) within \(\mathcal{O}^{(2)}\).

Given that \(d(\mathcal{O}^{(1)}, \partial \mathcal{O}^{(2)}) > \eta > 0\), the domain \(\mathcal{O}^{(1)}\) is compactly contained in \(\mathcal{O}^{(2)}\). We can thus select a finite collection of points \(\{z_j\}_{j=1}^N \subset \mathcal{O}^{(1)}\) and \(\{t_j\}_{j=1}^N \subset \mathbb{R}\) such that \(\mathcal{O}^{(1)}_{T} \subset \bigcup_{j=1}^N Q(R_j/2, t_j, z_j)\) and each \(Q(R_j, t_j, z_j) \subset \mathcal{O}^{(2)}_{T}\), where \(R_j = \gamma_j d(z_j)\) for each \(j\). This finite covering is feasible due to the relative compactness of \(\mathcal{O}^{(1)}\) in \(\mathcal{O}^{(2)}\) (note that \(N\) depends only on \(\mathcal{O}^{(1)}\), \(\mathcal{O}^{(2)}\), and \(\eta\)).

Let \((t_0, z_0)\) be an arbitrary point in \(\mathcal{O}^{(1)}\), and denote by \(\mathcal{L}_0\) the operator \(\mathcal{L}\) with coefficients frozen at \((t_0, z_0)\). Since the operator \((\partial_t - \mathcal{L} + r)\) is uniformly parabolic in \(D_T\), the constant-coefficient operator \((\partial_t - \mathcal{L}_0 + r)\) is also uniformly parabolic on \(\mathbb{R}^2\). Applying Proposition 7.18 and following the proof of Theorem 7.22 in \cite{Lieberman1996}, there exists a constant \(C\), depending only on \(\mathcal{O}^{(1)}\), \(\mathcal{O}^{(2)}\), \(\eta\), and on \(\gamma_j > 0\) chosen sufficiently small (depending on \(\mathcal{O}^{(1)}\), \(\mathcal{O}^{(2)}\), and \(p\)), such that
\[
R_j^2 \|D^2 v\|_{L_p(Q(R_j/2, t_j, z_j))} + R_j \|Dv\|_{L_p(Q(R_j/2, t_j, z_j))} \leq C \bigl( \|(\partial_t - \mathcal{L} + r)v\|_{L_p(Q(R_j, t_j ,z_j))} + \|v\|_{L_p(Q(R_j, t_j ,z_j))} \bigr).
\]
This yields
\begin{align*}
\|D^2 v\|_{L_p(Q(R_j/2, t_j, z_j))} &\leq \frac{C}{R_j^2} \Bigl( \|(\partial_t - \mathcal{L} + r)v\|_{L_p(Q(R_j, t_j, z_j))} + \|v\|_{L_p(Q(R_j, t_j, z_j))} \Bigr) \\
&= \frac{C}{R_j^2} \Bigl( \|f\|_{L_p(Q(R_j, t_j, z_j))} + \|v\|_{L_p(Q(R_j, t_j, z_j))} \Bigr).
\end{align*}
Raising both sides to the \(p\)-th power (for \(p \geq 3\)):
\begin{align*}
\|D^2 v\|_{L_p(\mathcal{O}^{(1)}_T)}^p
&= \int_{\mathcal{O}^{(1)}_T} |D^2 v|^p \, dz \, dt \\
&\leq \sum_{j=1}^N \int_{Q(R_j/2, t_j, z_j)} |D^2 v|^p \, dz \, dt \\
&= \sum_{j=1}^N \|D^2 v\|_{L_p(Q(R_j/2, t_j, z_j))}^p \\
&\leq \sum_{j=1}^N \left( \frac{C}{R_j^2} \bigl( \|f\|_{L_p(Q(R_j, t_j, z_j))} + \|v\|_{L_p(Q(R_j, t_j, z_j))} \bigr) \right)^p \\
&\leq 2^{p-1} \sum_{j=1}^N \left( \frac{C}{R_j^2} \right)^p
\Bigl( \|f\|_{L_p(Q(R_j, t_j, z_j))}^p + \|v\|_{L_p(Q(R_j, t_j, z_j))}^p \Bigr),
\end{align*}
where the last inequality follows from \((a + b)^p \leq 2^{p-1}(a^p + b^p)\) for \(a, b \geq 0\). Let \(\gamma = \min_{j=1}^N \gamma_j\). Since \(R_j = \gamma_j d(z_j) \geq \gamma \eta\), it holds that \(\frac{1}{R_j^2} \leq \frac{1}{(\gamma \eta)^2}\) for all \(j\), so
\[
\left( \frac{C}{R_j^2} \right)^p \leq \left( \frac{C}{(\gamma \eta)^2} \right)^p.
\]
Thus,
\[
\|D^2 v\|_{L_p(\mathcal{O}^{(1)}_T)}^p
\leq 2^{p-1} \left( \frac{C}{(\gamma \eta)^2} \right)^p
\left( \sum_{j=1}^N \|f\|_{L_p(Q(R_j, t_j, z_j))}^p + \sum_{j=1}^N \|v\|_{L_p(Q(R_j, t_j, z_j))}^p \right).
\]
Next, we have
\[
\sum_{j=1}^N \|f\|_{L_p(Q(R_j, t_j, z_j))}^p \leq N \|f\|_{L_p(\mathcal{O}^{(2)}_T)}^p,
\]
and similarly for the sum involving \(\|v\|_{L_p(Q(R_j, t_j, z_j))}^p\). Substituting these bounds gives
\[
\|D^2 v\|_{L_p(\mathcal{O}^{(1)}_T)}^p \leq 2^{p-1} \left( \frac{C}{(\gamma \eta)^2} \right)^p N
\left( \|f\|_{L_p(\mathcal{O}^{(2)}_T)}^p + \|v\|_{L_p(\mathcal{O}^{(2)}_T)}^p \right),
\]
which implies
\[
\|D^2 v\|_{L_p(\mathcal{O}^{(1)}_T)} \leq 2^{(p-1)/p} \frac{C}{(\gamma \eta)^2} N^{1/p}
\left( \|f\|_{L_p(\mathcal{O}^{(2)}_T)} + \|v\|_{L_p(\mathcal{O}^{(2)}_T)} \right).
\]
Equivalently, we can express this as
\begin{equation}
\label{eq:hessian-estimate}
\|D^2 v\|_{L_p(\mathcal{O}^{(1)}_T)} \leq C \left( \|f\|_{L_p(\mathcal{O}^{(2)}_T)} + \|v\|_{L_p(\mathcal{O}^{(2)}_T)} \right),
\end{equation}
where \(C\) depends only on \(p\), \(\mathcal{O}^{(1)}\), \(\mathcal{O}^{(2)}\), and \(\eta\). Following a similar procedure, we obtain
\begin{equation}
\label{eq:first-derivative-estimate}
\|Dv\|_{L_p(\mathcal{O}^{(1)}_T)} \leq C
\left( \|f\|_{L_p(\mathcal{O}^{(2)}_T)} + \|v\|_{L_p(\mathcal{O}^{(2)}_T)} \right).
\end{equation}

Next, we bound \(\|\partial_t v\|_{L_p(\mathcal{O}^{(1)}_T)}\). From the relation \(f = (\partial_t - \mathcal{L} + r)v\), we have \(\partial_t v = f + (\mathcal{L} - r)v\). Taking the \(L_p\) norm over \(\mathcal{O}^{(1)}_T\) and bounding the coefficients of \(\mathcal{L}\) by \(\lambda\), we find
\[
\|\partial_t v\|_{L_p(\mathcal{O}^{(1)}_T)} \leq
\lambda \|D^2 v\|_{L_p(\mathcal{O}^{(1)}_T)}
+ \lambda \|Dv\|_{L_p(\mathcal{O}^{(1)}_T)}
+ \|r v\|_{L_p(\mathcal{O}^{(1)}_T)}
+ \|f\|_{L_p(\mathcal{O}^{(1)}_T)}.
\]
Since \(\mathcal{O}^{(1)} \Subset \mathcal{O}^{(2)}\), it follows that
\(\|f\|_{L_p(\mathcal{O}^{(1)}_T)} \leq \|f\|_{L_p(\mathcal{O}^{(2)}_T)}\) and \(\|v\|_{L_p(\mathcal{O}^{(1)}_T)} \leq \|v\|_{L_p(\mathcal{O}^{(2)}_T)}\), yielding
\[
\|\partial_t v\|_{L_p(\mathcal{O}^{(1)}_T)} \leq
\lambda \|D^2 v\|_{L_p(\mathcal{O}^{(1)}_T)}
+ \lambda \|Dv\|_{L_p(\mathcal{O}^{(1)}_T)}
+ \|r v\|_{L_p(\mathcal{O}^{(2)}_T)}
+ \|f\|_{L_p(\mathcal{O}^{(2)}_T)}.
\]
Substituting the estimates from \eqref{eq:hessian-estimate} and \eqref{eq:first-derivative-estimate}, we arrive at
\[
\|\partial_t v\|_{L_p(\mathcal{O}^{(1)}_T)} \leq C
\left( \|f\|_{L_p(\mathcal{O}^{(2)}_T)} + \|v\|_{L_p(\mathcal{O}^{(2)}_T)} \right).
\]
Therefore, we conclude that
\[
\|v\|_{W^{2,1}_p(\mathcal{O}^{(1)}_T)} \leq C \bigl( \|f\|_{L_p(\mathcal{O}^{(2)}_T)} + \|v\|_{L_p(\mathcal{O}^{(2)}_T)} \bigr),
\]
where \(C\) depends only on \(p\), \(\mathcal{O}^{(1)}\), \(\mathcal{O}^{(2)}\), and \(\eta\).
\end{proof}

The estimates obtained above hold uniformly over bounded subdomains \(D \subset \mathbb{R} \times \mathbb{R}^+\) bounded away from \(y=0\). We now exploit this uniformity to approximate the full unbounded domain \(E_T\) by an expanding sequence of such subdomains and pass to the limit.

\subsection{Convergence to Viscosity Solution in the Unbounded Domain}\label{subsec:convergence}

\begin{definition}\label{def:strong-solution}
For each \(D_T\), a strong solution of the problem
\begin{equation}\label{eq:strong-solution-problem}
\begin{cases}
\min\{(\partial_t - \mathcal{L} + r)v, v - G\} = 0, & \text{ in } D_T, \\[0.5em]
v|_{\partial_P D_T} = g,
\end{cases}
\end{equation}
is a function \( v \in W^{2,1}_{3,\mathrm{loc}}(D_T) \cap C(\overline{D_T}) \) that satisfies \eqref{eq:strong-solution-problem} in \( D_T \) and attains the boundary data \( g \) on \( \partial_P D_T \).
\end{definition}

\begin{proposition}
\label{lem:strong-solution-existence}
Let \( D \subset \mathbb{R} \times \mathbb{R}^+ \) be a bounded open subset, bounded away from \(\{y = 0\}\), whose boundary belongs to class \(H^{\beta+2}\) for some \( \beta \in (0,1)\). Assume that the function $g$ satisfies the conditions stated in Lemma \ref{lem:g_epsilon_con}.
Then the following results hold:
\begin{itemize}
    \item[(i)] There exists a strong solution $v^*$ to the obstacle problem \eqref{eq:strong-solution-problem}.
\item[(ii)] \(v^*\) is a viscosity solution of problem \eqref{eq:strong-solution-problem}.
\end{itemize}
\end{proposition}

\begin{proof}
\begin{itemize}
 \item[(i)] Fix $\epsilon > 0$ and let $v^{\epsilon}$ denote the solution to the penalised problem \eqref{eq:penalty_equation}. By \cite[Theorem 9.1, p.~341]{Ladyzenskaja1968}, together with the uniform bound on \( p_\epsilon \left(v^{\epsilon} - G^{\epsilon}\right) \) mentioned in Proposition~\ref{lem:bounds_v_epsilon} and Lemma~\ref{lem:g_epsilon_con}, we have:
\begin{equation}
\label{eq:sobolev_bound}
\| v^{\epsilon} \|_{W_p^{2,1}(D_T)} \leq C,
\end{equation}
for a constant $C$ that does not depend on $\epsilon$. We now fix $p \geq 3$. The space $W_p^{2,1}(D_T)$ is reflexive, and so every bounded sequence admits a weakly convergent subsequence (cf.\ \citet[Appendix D.4]{evans1998}). In other words, we can extract a subsequence $(\epsilon_k)_{k \geq 0}$ with $\epsilon_k \to 0$ and a limit $v^* \in W_p^{2,1}(D_T)$ satisfying $v^{\epsilon_k} \rightharpoonup v^*$ weakly in $W_p^{2,1}(D_T)$. Now using the Sobolev embedding theorem (cf.\ \citet[Lemma 3.3, p.~80]{Ladyzenskaja1968}) and together with \eqref{eq:sobolev_bound} gives
\[
\|v^{\epsilon}\|_{D_T}^{(\beta)} \leq C, \quad \beta = 2 - \frac{4}{p},
\]
where $C$ is again independent of $\epsilon$. Because $p \geq 3$ ensures $\beta > 0$, the sequence $(v^{\epsilon_k})$ is uniformly bounded and equicontinuous on $D_T$. Hence we can apply the Arzelà–Ascoli theorem to extract a further subsequence, still written $(\epsilon_k)$ for notational simplicity, such that $v^{\epsilon_k} \to v^*$ uniformly on $D_T$. From the uniform limit of continuous functions, $v^*$ is itself continuous on $D_T$.

Let us show that $v^*$ solves Equation~\eqref{eq:strong-solution-problem} a.e.\ in \(D_T\). Since $p_{\epsilon_k}(v^{\epsilon_k} - G ^{\epsilon_k}) \leq 0$, we have $(\partial_t - \mathcal{L} + r)v^{\epsilon_k} \geq 0$ for each $\epsilon_k$. Then, utilising the dominated convergence theorem, we obtain
\[
\int (\partial_t - \mathcal{L} + r)v^* \phi \, dzdt = \lim_{\epsilon_k \to 0} \int (\partial_t - \mathcal{L} + r)v^{\epsilon_k} \phi \, dzdt \geq 0
\]
for any compactly supported smooth function $\phi$. The previous inequality then yields $(\partial_t - \mathcal{L} + r)v^* \geq 0$ on $D_{T}$ in the distributional sense. On the other hand, Proposition~\ref{lem:bounds_v_epsilon} shows that $v^{\epsilon_k} \geq G^{\epsilon_k}$, hence $v^* \geq G$ after sending $\epsilon_k \to 0$. Therefore, we obtain $\min\{(\partial_t - \mathcal{L} + r)v^*, v^* - G\} \geq 0$ on $D_{T}$ in the distributional sense.

We now prove that $(\partial_t - \mathcal{L} + r)v^* = 0$ when $v^* > G$. Take any $(t,x,y) \in D_T$ such that $v^*(t,x,y) > G(x)$. Since $v^*$ and $G$ are continuous, we can find a small $\delta > 0$ such that $v^*(\tilde{t},\tilde{x},\tilde{y}) > G(\tilde{x}) + 2\delta$ for any $(\tilde{t},\tilde{x},\tilde{y})$ inside a neighbourhood of $(t,x,y)$. Utilising the uniform convergence of $(v^{\epsilon_k})_{k>0}$ and $(G^{\epsilon_k})_{k>0}$, we can find a small enough $\epsilon_k$ such that $v^{\epsilon_k}(\tilde{t},\tilde{x},\tilde{y}) > G^{\epsilon_k}(\tilde{x}) + \delta$ in the aforementioned neighbourhood. For all $\epsilon_k < \delta$, we have $v^{\epsilon_k} - G^{\epsilon_k} > \delta > \epsilon_k$ in this neighbourhood. By \eqref{eq:43_hxing}(ii), $p_{\epsilon_k}(\xi) = 0$ for $\xi \geq \epsilon_k$, so $p_{\epsilon_k}(v^{\epsilon_k} - G^{\epsilon_k}) = 0$ and therefore $(\partial_t - \mathcal{L} + r)v^{\epsilon_k} = 0$ in this neighbourhood. After sending $\epsilon_k \to 0$, we obtain $(\partial_t - \mathcal{L} + r)v^* = 0$ in the distributional sense when $v^* > G$. Also from \eqref{eq:sobolev_bound}, we easily deduce that $v^* \in W^{2,1}_{p,\mathrm{loc}}(D_T)$ for any \(p \geq 3\), hence $ v^*$ also solves Equation~\eqref{eq:strong-solution-problem} a.e.\ in $D_T$.

We finally prove that \(v^* \in C(\overline{D_T})\). The argument is similar to the proof in Proposition \ref{lem:penalty_var}. Consider the subsequence \((\epsilon_j)_{j \in \mathbb{N}}\) along which \(v^{\epsilon_j} \to v^*\) pointwise in \(\overline{D_T}\). Let \(\bar{\xi} = (\bar{t}, \bar{x},\bar{y}) \in \partial_P D_T\) and \(\alpha > 0\) be given. There exists an open neighbourhood \(V\) of \(\bar{\xi}\) such that
\[
|g(x) - g(\bar{x})| \leq \alpha \quad \text{for all } \xi = (t,x,y) \in V \cap \partial_P D_T.
\]
From Proposition~\ref{lem: barrier_existence}, there exists a barrier function \(w\) for the operator \(\partial_t - \mathcal{L} + r\) in \(V \cap D_T\). Define the auxiliary functions
\[
v^\pm(\xi) = g(\bar{x}) \pm (\alpha + k_\alpha w(\xi)),
\]
where \(k_\alpha > 0\) is a sufficiently large constant, independent of \(j\). By Proposition~\ref{lem:bounds_v_epsilon}, the penalty term \(p_{\epsilon}(v^{\epsilon} - G^{\epsilon})\) is uniformly bounded in \(\epsilon\). Therefore we can choose \(k_\alpha\) large enough to ensure
\[
-(\partial_t - \mathcal{L} + r)(v^{\epsilon_j} - v^+) \geq p_{\epsilon_j}(v^{\epsilon_j} - G^{\epsilon_j}) + k_\alpha \geq 0 \quad \text{in } V \cap D_T.
\]
We now show that $v^{\epsilon_j} \leq v^+$ on $\partial(V \cap D_T)$. This boundary decomposes into two parts. On $V \cap \partial_P D_T$, since $g^{\epsilon_j} \to g$ uniformly, for all $j$ sufficiently large we have $|g^{\epsilon_j}(x) - g(x)| \leq \alpha/2$. By shrinking $V$ if necessary, we may also assume $|g(x) - g(\bar{x})| \leq \alpha/2$ on $V \cap \partial_P D_T$. Therefore $v^{\epsilon_j} = g^{\epsilon_j}(x) \leq g(\bar{x}) + \alpha \leq v^+$ on this part. On $\partial V \cap \overline{D_T}$, since these points are bounded away from $\bar{\xi}$, property (ii) in Definition \ref{def:barrier} gives $w(\xi) \geq \delta$ for some $\delta > 0$. Since $v^{\epsilon_j} \leq K + 1$ by Proposition~\ref{lem:bound_epsilon}~(b), choosing $k_\alpha \geq (K + 1 - g(\bar{x}) - \alpha)/\delta$ ensures $v^+(\xi) \geq g(\bar{x}) + \alpha + k_\alpha \delta \geq K + 1 \geq v^{\epsilon_j}$.

Lemma~\ref{lem:nonnegative_v} then implies \(v^{\epsilon_j} \leq v^+\) in \(V \cap D_T\). An analogous construction yields \(v^{\epsilon_j} \geq v^-\) in \(V \cap D_T\). Passing to the limit as \(j \to \infty\), we obtain
\[
g(\bar{x}) - \alpha - k_\alpha w(\xi) \leq v^*(\xi) \leq g(\bar{x}) + \alpha + k_\alpha w(\xi), \quad \xi \in V \cap D_T.
\]
Consequently,
\[
g(\bar{x}) - \alpha \leq \liminf_{\xi \to \bar{\xi}} v^*(\xi) \leq \limsup_{\xi \to \bar{\xi}} v^*(\xi) \leq g(\bar{x}) + \alpha, \quad \xi \in V \cap D_T.
\]
Since \(\alpha > 0\) is arbitrary, we have \(v^*(\xi) \to g(\bar{x})\) when \(\xi \to \bar{\xi}\) with \(\xi \in D_T\). Combined with the uniform continuity of \(v^*\) in the interior, this establishes \(v^* \in C(\overline{D_T})\).

\item[(ii)]

We show that \(v^*\) is a viscosity subsolution of \eqref{eq:strong-solution-problem}. Fix \( (t, x, y) \in D_T \), and assume \( v^*(t, x, y) > G(x) \); otherwise, the condition holds trivially. Let \( \phi \in C^{1,2}([0, T] \times \mathbb{R} \times \mathbb{R}^+) \) be a test function such that \( v^* - \phi \) has a strict maximum at \( (t, x, y) \) in a neighbourhood \( B(t, x, y; \delta) \subset D_T\). For each \(k\), there exists a point \( (t^{(k)}, x^{(k)}, y^{(k)}) \in B(t, x, y; \delta) \) such that:
    \[
    (v^{\epsilon_k} - \phi)(t^{(k)}, x^{(k)}, y^{(k)}) \text{ is a maximum over } B(t, x, y; \delta).
    \]
    Let \( (t^*, x^*, y^*) \in B(t, x, y; \delta) \) be the limit of \( (t^{(k)}, x^{(k)}, y^{(k)}) \) as \( k \to \infty \). For any \( (t', x', y') \in B(t, x, y; \delta) \):
    \[
    (v^{\epsilon_k} - \phi)(t^{(k)}, x^{(k)}, y^{(k)}) \geq (v^{\epsilon_k} - \phi)(t', x', y').
    \]
    Taking the limit as \( k \to \infty \), uniform convergence gives
    \[
    (v^* - \phi)(t^*, x^*, y^*) \geq (v^* - \phi)(t', x', y').
    \]
    Since \( (t, x, y) \) is a strict maximum, \( (t^*, x^*, y^*) = (t, x, y) \), so \( (t^{(k)}, x^{(k)}, y^{(k)}) \to (t, x, y) \). Because \( v^{\epsilon_k} \) is a classical solution of Equation~\eqref{eq:penalty_equation}, it is also a viscosity solution; hence
    \[
    (\partial_t - \mathcal{L} + r) \phi(t^{(k)}, x^{(k)}, y^{(k)}) + p_{\epsilon_{k}}(v^{\epsilon_k} - G^{\epsilon_k})\leq 0.
    \]

Now, since \(v^*(t, x,y) > G(x)\) and \(v^{\epsilon_k}(t^{(k)}, x^{(k)}, y^{(k)}) - G^{\epsilon_k}(x^{(k)})\) converges to \(v^*(t, x,y) - G(x)\), we obtain \(\lim_{\epsilon_k \to 0} p_{\epsilon_k}(v^{\epsilon_k}(t^{(k)}, x^{(k)}, y^{(k)}) - G^{\epsilon_k}(x^{(k)})) = 0\). As a result, \((\partial_t - \mathcal{L} + r)\phi(t, x, y) \leq 0\) by sending \(\epsilon_k \to 0\). This confirms that \(v^*\) is a viscosity subsolution of Equation~\eqref{eq:strong-solution-problem}.

For the supersolution property, let \(\phi \in C^{1,2}([0,T] \times \mathbb{R} \times \mathbb{R}^+)\) be a test function such that \(v^* - \phi\) has a minimum at \((t,x,y) \in D_T\). Since \(v^* \geq G\), the condition \(v^*(t,x,y) - G(x) \geq 0\) holds automatically. It remains to show \((\partial_t - \mathcal{L} + r)\phi(t,x,y) \geq 0\). By the same covering argument as for the subsolution, there exist points \((t^{(k)}, x^{(k)}, y^{(k)}) \to (t,x,y)\) at which \(v^{\epsilon_k} - \phi\) attains a minimum. Since \(v^{\epsilon_k}\) is a classical solution of \eqref{eq:penalty_equation}, we have \((\partial_t - \mathcal{L} + r)\phi(t^{(k)}, x^{(k)}, y^{(k)}) + p_{\epsilon_k}(v^{\epsilon_k} - G^{\epsilon_k}) \geq 0\). Since \(p_{\epsilon_k} \leq 0\), this gives \((\partial_t - \mathcal{L} + r)\phi(t^{(k)}, x^{(k)}, y^{(k)}) \geq 0\). Sending \(k \to \infty\) yields \((\partial_t - \mathcal{L} + r)\phi(t,x,y) \geq 0\), completing the supersolution verification.
\end{itemize}
\end{proof}

\begin{proposition}
\label{lem:subsequence_existence}
For every \(p \ge 3\), there exists a function \(v^* \in W^{2,1}_{p,\mathrm{loc}}(E_T) \cap C([0,T] \times \mathbb{R} \times \mathbb{R}^+)\) that is a viscosity solution of Equation \eqref{eq:heston_variational_inequality}.
\end{proposition}

\begin{proof}

We provide the proof in two steps. In Step 1, we will construct a function \(v^* \in W^{2,1}_{p,\mathrm{loc}}(E_T)\) using the results from previous propositions; then, in Step 2, we will prove that \(v^*\) is a viscosity solution of Equation \eqref{eq:heston_variational_inequality}.

\textbf{Step 1:} We establish the existence of a function \(v^* \in W^{2,1}_{p,\mathrm{loc}}(E_T)\) by considering a sequence of obstacle problems posed on expanding cylindrical domains \(D_T^n \subset [0,T] \times \mathbb{R} \times \mathbb{R}^+\), defined as the subset of \([0,T] \times \{|x| < n\} \times \{ \frac{1}{n} < y < n\}\).
Each domain \(D_T^n\) is chosen such that its parabolic boundary \(\partial_P D_T^n\) is of class \(H^{\beta+2}\) for some \(\beta \in (0,1)\), with \(D_T^n \subset D_T^{n+1}\) and \(\bigcup_{n \in \mathbb{N}} D_T^n = E_T\). Define the approximate boundary data
\[
g_n(x) = \frac{1}{2} \left( (K - e^{x}) + \sqrt{(K - e^{x})^2 + \frac{1}{n}} \right).
\]
It is straightforward to verify that \(g_n\) satisfies the conditions of Lemma~\ref{lem:g_epsilon_con}, with the constant \(M > 0\) independent of \(n\). Moreover, \(g_n(x) \to G(x)\) pointwise as \(n \to \infty\). By Proposition~\ref{lem:strong-solution-existence}, for each \(n \in \mathbb{N}\) there exists a strong solution \(v_n\) of the obstacle problem
\begin{equation}
\label{eq:vn_variational}
\begin{cases}
\min\bigl\{(\partial_t - \mathcal{L} + r)v_n, \, v_n - G\bigr\} = 0, & \text{in } D_T^n, \\
v_n|_{\partial_P D_T^n} = g_n.
\end{cases}
\end{equation}

We now show that the sequence \((v_n)\) is decreasing. Suppose, for contradiction, that the set \(B = \{z \in D_T^n : v_{n+1}(z) > v_n(z)\}\) is non-empty (where \(v_{n+1}\) is extended by \(g_{n+1}\) outside \(D_T^n\)). In \(B\) we have \(v_{n+1} > v_n \geq G\), so
\[
(\partial_t - \mathcal{L} + r)v_{n+1} = 0, \qquad (\partial_t - \mathcal{L} + r)v_n \geq 0.
\]
Thus,
\[
(\partial_t - \mathcal{L} + r)(v_{n+1} - v_n) \leq 0 \quad \text{in } B,
\]
while \(v_{n+1} - v_n = 0\) on \(\partial_P B \). From Corollary 7.4 in \cite{Lieberman1996}, we obtain \(v_{n+1} - v_n \leq 0\) in \(\overline{B}\), contradicting the definition of \(B\). Hence, \(v_{n+1} \leq v_n\) in \(D_T^n\). Combining this monotonicity with the uniform bounds from Proposition~\ref{lem:bounds_v_epsilon}, we obtain
\[
G(x) \leq v_{n+1}(t,x,y) \leq v_n(t,x,y) \leq K + 1 \quad \text{in } D_T^n.
\]
By the monotone convergence theorem, the pointwise limit
\[
v^*(t,x,y) = \lim_{n \to \infty} v_n(t,x,y)
\]
exists on \(E_T\), and \(G \leq v^* \leq K + 1\). To establish that the limit function \(v^*\) belongs to \(W^{2,1}_{p,\mathrm{loc}}(E_T)\), fix an arbitrary compact set \(\mathcal{O} \Subset \mathbb{R} \times \mathbb{R}^+\). We define
\begin{align*}
y_1 &= \inf \{ y : (x, y) \in \mathcal{O} \}, & y_2 &= \sup \{ y : (x, y) \in \mathcal{O} \}, \\
x_1 &= \inf \{ x : (x, y) \in \mathcal{O} \}, & x_2 &= \sup \{ x : (x, y) \in \mathcal{O} \}.
\end{align*}
The expanded domain \(\mathcal{O}^{\gamma}\) (\(\gamma > 0\)) is given by
\[
\mathcal{O}^{\gamma} = \left( x_1 - \gamma, x_2 + \gamma \right) \times \left( y_1- \gamma, y_2 + \gamma \right).
\]
Now we choose \(\bar{\gamma} > 0\) sufficiently small (note that \(\bar{\gamma}\) depends only on \(\mathcal{O}\)) such that \(\mathcal{O}^{\bar{\gamma}} \Subset \mathbb{R} \times \mathbb{R}^+\). Then we choose \(n_0 \in \mathbb{N}\) such that \(\mathcal{O}^{\bar{\gamma}}_T \Subset D_T^n\) for all \(n \geq n_0\). The $W^{2,1}_p$ estimates from Proposition~\ref{lem:cauchy_problem_solution}, combined with the boundedness established in Proposition~\ref{lem:bounds_v_epsilon}, imply
\[
\|v_n\|_{W^{2,1}_p(\mathcal{O}_T)} \leq C_{\mathcal{O},p, \bar{\gamma}}
\]
for all \(n \geq n_0\), where \(C_{\mathcal{O},p, \bar{\gamma}}\) depends only on \(\mathcal{O}\), \(p\) and \(\bar{\gamma}\) (not on \(n\)). Passing to the limit \(n \to \infty\), we conclude that \(v^* \in W^{2,1}_{p,\mathrm{loc}}(E_T)\).

    \textbf{Step 2:}
    We will prove \( v^* \) is a viscosity solution if it is both a supersolution and a subsolution.

    \textbf{Subsolution}: Fix \( (t, x, y) \in (0, T] \times \mathbb{R} \times \mathbb{R}^+ \), and assume \( v^*(t, x, y) > G(x) \); otherwise, the condition holds trivially. Let \( \phi \in C^{1,2}([0, T] \times \mathbb{R} \times \mathbb{R}^+) \) be a test function such that \( v^* - \phi \) has a strict maximum at \( (t, x, y) \) in a neighbourhood \( B(t, x, y; \delta) \subset (0, T] \times \mathbb{R} \times \mathbb{R}^+ \). Choose \( n \) such that \( B(t, x, y; \delta) \subset D^n_T \). Since \( v_n \to v^* \) uniformly on compact subsets of \(E_T\), there exists a point \( (t^{(n)}, x^{(n)}, y^{(n)}) \in B(t, x, y; \delta) \) such that:
    \[
    (v_n - \phi)(t^{(n)}, x^{(n)}, y^{(n)}) \text{ is a maximum over } B(t, x, y; \delta).
    \]
    Following arguments similar to those employed in the proof of Proposition~\ref{lem:strong-solution-existence}, we deduce that $(t^{(n)}, x^{(n)}, y^{(n)}) \to (t, x, y)$. By Proposition~\ref{lem:strong-solution-existence}, \( v_n \) is also a viscosity solution of \eqref{eq:strong-solution-problem}; hence
    \[
    (\partial_t - \mathcal{L} + r) \phi(t^{(n)}, x^{(n)}, y^{(n)}) \leq 0.
    \]
    Taking the limit as \( n \to \infty \), we obtain
    \[
    (\partial_t - \mathcal{L} + r) \phi(t, x, y) \leq 0.
    \]
    At \( t = 0 \), \( v^*(0, x, y) = \lim_{n \to \infty} g_n(x) = G( x) \), satisfying the initial condition. Therefore, we conclude that \(v^*\) is a viscosity subsolution of Equation \eqref{eq:heston_variational_inequality}.

\textbf{Supersolution}: Let \(\phi \in C^{1,2}([0,T] \times \mathbb{R} \times \mathbb{R}^+)\) be a test function such that \(v^* - \phi\) has a minimum at \((t,x,y) \in (0,T] \times \mathbb{R} \times \mathbb{R}^+\). Since \(v^* \geq G\) (from Proposition~\ref{lem:bounds_v_epsilon}), we have \(v^*(t,x,y) - G(x) \geq 0\). It remains to show \((\partial_t - \mathcal{L} + r)\phi(t,x,y) \geq 0\). Choose \(n\) such that \((t,x,y) \in D^n_T\). By the same argument as for the subsolution, there exist points \((t^{(n)}, x^{(n)}, y^{(n)}) \to (t,x,y)\) at which \(v_n - \phi\) attains a minimum. Since \(v_n\) is a viscosity supersolution of \eqref{eq:strong-solution-problem}, we have \((\partial_t - \mathcal{L} + r)\phi(t^{(n)}, x^{(n)}, y^{(n)}) \geq 0\). Sending \(n \to \infty\) yields \((\partial_t - \mathcal{L} + r)\phi(t,x,y) \geq 0\). At \(t = 0\), \(v^*(0,x,y) = \lim_{n \to \infty} g_n(x) = G(x)\), so the initial condition holds. Thus, \(v^*\) is a viscosity solution of Equation~\eqref{eq:heston_variational_inequality}.
\end{proof}

\begin{remark}
In this proof, a key idea is the selection of the function \( g_n \), which converges to \( G \) and satisfies the conditions of Lemma~\ref{lem:g_epsilon_con}. Since \( \partial_P D^n_T \neq \partial_P D^{n+1}_T \), setting the boundary condition as \( v_n|_{\partial_P D^n_T} = G|_{\partial_P D^n_T} \) for each \( n \) prevents convergence to a consistent limit \( v^* \). The choice of \( g_n \) overcomes this issue. Furthermore, by combining Proposition~\ref{lem:bounds_v_epsilon} and \ref{lem:cauchy_problem_solution}, we uniformly bound \( \| v_n \|_{W_p^{2,1}(\mathcal{O}_T)} \), independent of \( n \), for any compact \( \mathcal{O} \subset \mathbb{R} \times \mathbb{R}^+ \). Consequently, the limit \(v^* \in W^{2,1}_{p,\mathrm{loc}}(E_T)\).
\end{remark}

\subsection{Derivation of the smooth-fit principle and $C^{1,2}$-regularity in continuation region}\label{subsec:smoothfit}

We now have all ingredients in hand and proceed to prove the main results.

\begin{proposition}
\label{lem:smoothfits}
The derivatives \(\partial_x P^A\) and \(\partial_y P^A\) are continuous on \([0,T) \times \mathbb{R} \times \mathbb{R}^+\).
        (\textit{The smooth-fit holds})
\end{proposition}
\begin{proof}
Define \(u^*(t,x,y) = v^*(T - t, x, y)\). Then \(u^*\) satisfies Equation~\eqref{eq:heston_pde} in the viscosity sense. Since both \(P^A\) and \(u^*\) are bounded and continuous on \([0,T] \times \mathbb{R} \times \mathbb{R}^+\), the uniqueness result for viscosity solutions established in Proposition~\ref{lem:unique_viscosity_solution2} yields
\[
u^*(t,x,y) = P^A(t,x,y).
\]
Let \(\mathcal{O} \subset \mathbb{R} \times \mathbb{R}^+\) be an arbitrary compact set. From Proposition~\ref{lem:subsequence_existence}, we have \(v^* \in W^{2,1}_{p,\mathrm{loc}}(E_T)\) for every \(p \geq 3\). Consequently, \(u^* \in W^{2,1}_{p,\mathrm{loc}}([0,T) \times \mathbb{R} \times \mathbb{R}^+)\). Applying the Sobolev embedding theorem (cf.~Lemma~3.3 in~\cite{Ladyzenskaja1968}, p.~80) on the cylinder \(\mathcal{O} \times [0,T-s]\) with \(s > 0\) arbitrary, we obtain
\[
P^A \in H^{\beta,\beta/2}(\mathcal{O} \times [0,T-s]),
\]
where \(\beta = 2 - \frac{4}{p}\) and \(s < T\). Choosing \(p > 4\) yields \(\beta > 1\), so that the Hölder space \(H^{\beta, \beta/2}\) contains continuous first-order spatial derivatives. Thus, \(\partial_x P^A\) and \(\partial_y P^A\) are continuous on \(\mathcal{O} \times [0,T-s]\). Since \(\mathcal{O}\) and \(s > 0\) are arbitrary, the derivatives \(\partial_x P^A\) and \(\partial_y P^A\) are continuous on \([0,T) \times \mathbb{R} \times \mathbb{R}^+\), establishing the smooth-fit principle.
\end{proof}

\begin{proposition}
\label{lem:c12}
The value function of the American put option, \(P^A \), exhibits \( C^{1,2} \) regularity in the continuation region.
\end{proposition}

\begin{proof}
For $B \Subset \mathbb{R} \times \mathbb{R}^+$ and $t_1, t_2 \in [0,T]$ such that $B \times (t_1, t_2) \subset \mathcal{C}$, consider the following boundary value problem:
\begin{equation}
\label{eq:continous_equation}
\begin{cases}
(-\partial_t - \mathcal{L} + r) v = 0, & (t,x,y) \in B \times [t_1, t_2), \\
v(t,x,y) = P^A(t,x,y), & (t,x, y) \in \partial B \times [t_1, t_2] \cup \bar{B} \times \{t_2\}.
\end{cases}
\end{equation}
We verify that $P^A$ is a viscosity solution of $(-\partial_t - \mathcal{L} + r)v = 0$ on $B \times (t_1, t_2)$.
For the subsolution property, let $\varphi$ be a test function such that $P^A - \varphi$ attains a maximum at $(t_0, x_0, y_0) \in B \times (t_1, t_2)$. Since $P^A$ is a viscosity subsolution of \eqref{eq:heston_pde}, we have
\[
\min\bigl\{(-\partial_t - \mathcal{L} + r)\varphi(t_0, x_0, y_0),\, P^A(t_0, x_0, y_0) - G(x_0)\bigr\} \leq 0.
\]
Because $(t_0, x_0, y_0) \in \mathcal{C}$, the second argument satisfies $P^A(t_0, x_0, y_0) - G(x_0) > 0$, so the inequality must come from the first argument: $(-\partial_t - \mathcal{L} + r)\varphi(t_0, x_0, y_0) \leq 0$.

For the supersolution property, let $\varphi$ be a test function such that $P^A - \varphi$ attains a minimum at $(t_0, x_0, y_0) \in B \times (t_1, t_2)$. The supersolution condition for \eqref{eq:heston_pde} gives
\[
\min\bigl\{(-\partial_t - \mathcal{L} + r)\varphi(t_0, x_0, y_0),\, P^A(t_0, x_0, y_0) - G(x_0)\bigr\} \geq 0,
\]
which requires both arguments to be nonnegative; in particular, $(-\partial_t - \mathcal{L} + r)\varphi(t_0, x_0, y_0) \geq 0$.

Combining the two, $P^A$ is a viscosity solution of $(-\partial_t - \mathcal{L} + r)v = 0$ on $B \times (t_1, t_2)$. The boundary condition is satisfied trivially since the prescribed data is $P^A$ itself.

Since the boundary and terminal values of \eqref{eq:continous_equation} are continuous, Theorem~9 in \cite{Friedman1964} provides a classical solution $u^* \in C^{1,2}(B \times (t_1, t_2))$. Every classical solution is in particular a viscosity solution, so both $P^A$ and $u^*$ are viscosity solutions of \eqref{eq:continous_equation} with the same boundary data. Since $B \Subset \mathbb{R} \times \mathbb{R}^+$, the operator $\mathcal{L}$ is uniformly parabolic with Lipschitz continuous coefficients on $\bar{B}$. The comparison principle for viscosity solutions (Theorem~8.2 in \cite{Crandall1992}), applied to the pair $(P^A, u^*)$ and then to $(u^*, P^A)$, yields $P^A = u^*$ on $\bar{B} \times [t_1, t_2]$. Hence $P^A \in C^{1,2}(B \times (t_1, t_2))$. Since $B \times (t_1, t_2)$ is an arbitrary relatively compact subset of $\mathcal{C}$, the statement follows.
\end{proof}

\bibliographystyle{apalike}
\bibliography{references}

\appendix
\renewcommand{\theequation}{\Alph{section}.\arabic{equation}}



\end{document}